\documentclass{article}

\usepackage{amsfonts}
\usepackage{tikz}
\usetikzlibrary{backgrounds,automata}
\usetikzlibrary{shapes.geometric}

\newcommand{\R}{\mathbb{R}}
\newcommand{\E}{\operatorname{\mathbb{E}}}

\newcommand{\F}{\mathcal{F}}

\newcommand{\trho}{\tilde{\rho}}

\newcommand{\be}{\begin{equation}}
\newcommand{\ee}{\end{equation}}
\newcommand{\bea}{\begin{eqnarray}}
\newcommand{\eea}{\end{eqnarray}}
\newcommand{\beas}{\begin{eqnarray*}}
\newcommand{\eeas}{\end{eqnarray*}}

\usepackage{amsmath}
\usepackage{graphicx}
\usepackage{multirow}
\usepackage{color}
\usepackage[square,sort,comma,numbers]{natbib}
\RequirePackage[colorlinks,citecolor=blue,urlcolor=blue]{hyperref}
\usepackage{amssymb}
\usepackage[british]{babel}
\usepackage{subfigure}

\usepackage[utf8]{inputenc}

\usepackage{amsthm}
\usepackage{latexsym}
\usepackage{mathrsfs}
\usepackage[mathscr]{euscript}
\usepackage{array}

\newtheorem{theorem}{Theorem}

\newtheorem{corollary}[theorem]{Corollary}

\newtheorem{definition}[theorem]{Definition}
\newtheorem{example}[theorem]{Example}

\newtheorem{lemma}[theorem]{Lemma}

\newtheorem{proposition}[theorem]{Proposition}
\newtheorem{remark}[theorem]{Remark}

\begin{document}





\title{
Quasi-Logconvex Measures of Risk\thanks{We are very grateful to Fabio Bellini, Freddy Delbaen and Marco Frittelli
and to conference and seminar participants at the Model Uncertainty and Robust Finance Workshop in Milano 
and the Technical University of Munich 
for their comments and suggestions.
This research was funded in part by the Netherlands Organization for Scientific Research under grants NWO-VIDI and NWO-VICI (Laeven).}} 
\author{Roger J. A. Laeven \\
{\footnotesize Dept. of Quantitative Economics}\\
{\footnotesize University of Amsterdam, CentER}\\
{\footnotesize and EURANDOM}\\
{\footnotesize \texttt{R.J.A.Laeven@uva.nl}}\\
\and Emanuela Rosazza Gianin \\
{\footnotesize Dept. of Statistics and Quantitative Methods}\\
{\footnotesize University of Milano Bicocca}\\
{\footnotesize \texttt{emanuela.rosazza1@unimib.it}}\\
[1cm] }
\date{First Version: 5th July 2019.\\ This Version: \today. }

\maketitle

\begin{abstract}
This paper introduces and fully characterizes
the novel class of quasi-logconvex measures of risk,
to stand on equal footing with the rich class of quasi-convex measures of risk.
Quasi-logconvex risk measures naturally generalize logconvex return risk measures,
just like quasi-convex risk measures generalize convex monetary risk measures.
We establish their dual representation and analyze their taxonomy in a few (sub)classification results.
Furthermore, we characterize quasi-logconvex risk measures in terms of properties of families of acceptance sets
and provide their law-invariant representation.
Examples and applications to portfolio choice and capital allocation are also discussed.\\[3mm]
\noindent \textbf{Keywords:} Quasi-convex risk measures;
Robustness;
Monetary and convex risk measures;
Return risk measures;
Logconvexity;
Star-shapedness;
Liquidity;
Diversification.\\[3mm]
\noindent \textbf{AMS 2010 Classification:} Primary: 91B06, 91B30; Secondary: 62P05, 90C46.\\[3mm]
\noindent \textbf{JEL Classification:} D81, G10, G20.
\end{abstract}

%
%
%
%

\section{Introduction}

Nowadays a rich and sophisticated literature exists on the theory of measures of risk
and their applications to risk capital calculations, portfolio choice, valuation, and hedging.
Modern classes of risk measures are given by
coherent (Artzner, Delbaen, Eber and Heath \cite{ADEH99}),
convex (F\"ollmer and Schied \cite{FS02}, Frittelli and Rosazza Gianin \cite{FR02}),
entropy coherent and entropy convex (Laeven and Stadje \cite{LS13}),
monetary (F\"ollmer and Schied \cite{FS11}, Delbaen \cite{D12}),
return (Bellini, Laeven and Rosazza Gianin \cite{BLR18}),
cash-subadditive (El Karoui and Ravanelli \cite{ELKR09}),
and quasi-convex (Cerreia-Vioglio, Maccheroni, Marinacci and Montrucchio \cite{CMMM11}, Drapeau and Kupper \cite{DK13}, Frittelli and Maggis \cite{FM11})
measures of risk,
with vast roots in decision theory, actuarial and financial mathematics, and operations research.

The class of coherent risk measures is encompassed by the wider class of convex risk measures.
Initially, a key motivation behind the 
generalization of coherent risk measures
to the class of convex risk measures was the restrictiveness of the
so-called \textit{positive homogeneity} condition 
that coherent risk measures assume 
--- a risk measure $\rho$ is positively homogeneous if $\rho\left(\lambda X\right)=\lambda\rho(X)$, $\lambda\geq 0$.
It was argued that, due to liquidity risk, this assumption is easily violated.
Indeed, in many cases the weaker condition $\rho\left(\lambda X\right)\geq \lambda\rho(X)$, $\lambda\geq 1$,
often referred to as \textit{star-shapedness} or \textit{subhomogeneity}, may be a more natural requirement.  
Any convex 
risk measure is also star-shaped (Frittelli and Rosazza Gianin \cite{FR02}, Remark 8).

Coherent, convex, and entropy convex risk measures occur as special cases of monetary risk measures,
and all rely on the so-called \textit{translation invariance} condition ---
$\rho$ is translation invariant if $\rho\left(X+h\right)=\rho(X)+h$, $h \in \mathbb{R}$ --- as a basic axiom.
For a decade or so, the translation invariance requirement
was accepted as a natural condition for a capital requirement.
Indeed, with the interpretation of the minimal amount of cash to be added to a financial position
to make it acceptable from a regulatory perspective,
necessarily $\rho\left(X-\rho(X)\right)=0$,
which is implied by translation invariance.
Monetary risk measures also occur naturally in applications such as valuation, hedging, and risk sharing;
see e.g., El Karoui and Quenez \cite{ELKQ97}, Carr, Geman and Madan \cite{CGM01}, Filipovic and Kupper \cite{FK08},
Laeven and Stadje \cite{LS14}, Kr\"atschmer \textit{et al.} \cite{KLLSS18}, and Jouini, Schachermayer and Touzi \cite{JST08}.

In recent years, the translation invariance requirement has been subject to debate.
In an important contribution, El Karoui and Ravanelli \cite{ELKR09} argue that a risk measure should instead be \textit{cash-subadditive} ---
$\rho$ is cash-subadditive if $\rho\left(X+h\right)\leq\rho(X)+h$, $h \in \mathbb{R}_{+}$ ---
in the presence of interest rate risk, 
and they establish a representation of cash-subadditive convex measures of risk.
In addition, an insightful paper by Cerreia-Vioglio, Maccheroni, Marinacci and Montrucchio \cite{CMMM11} shows
that if translation invariance is replaced by cash-subadditivity,
then convexity should be replaced by \textit{quasi-convexity} ---
$\rho\left(\alpha X +(1-\alpha)Y\right) \leq \max\{\rho\left(X\right),\rho\left(Y\right)\}$, $\alpha \in (0,1)$ ---,
and provides representation results for this general class of risk measures.
Applications of quasi-convex risk measures can be found in e.g., Mastrogiacomo and Rosazza Gianin \cite{MRG15}.

From a different angle, Bellini, Laeven and Rosazza Gianin \cite{BLR18} argue that,
while the translation invariance requirement can be, and usually is, assumed at the level of absolute positions in monetary values
leading to the class of monetary risk measures,
it can also be naturally assumed on relative positions in (log) returns
leading to the class of return risk measures.
Return risk measures are monotone and positively homogeneous.
They constitute the relative counterpart of monetary risk measures,
much like the relative risk aversion function is the relative counterpart of the absolute risk aversion function
in measuring risk aversion (Pratt \cite{P64}).
Bellini, Laeven and Rosazza Gianin \cite{BLR18} establish a one-to-one correspondence between monetary and return risk measures,
provide representation results for (sub)classes of return risk measures,
and give several examples of return risk measures.
Bellini, Laeven and Rosazza Gianin \cite{BLR21} analyze dynamic return risk measures and characterize their time-consistency.

A potential limitation of return risk measures is the positive homogeneity condition 
they satisfy. 
In this paper, we relax this condition by introducing and characterizing
the 
class of 
quasi-logconvex and, optionally, star-shaped measures of risk.
We explicate that quasi-logconvex measures of risk stand on equal footing with the rich class of quasi-convex measures of risk.
We show that 
quasi-logconvex risk measures take the following form on $L^{\infty}_{++}$:
\begin{equation}
\rho(X)=\sup_{Q\in\mathcal{Q}}R\left(H_{0,Q}\left(X\right); Q\right),
\label{eq:repQLC}
\end{equation}
where $R:\mathbb{R}_{++}\times\mathcal{Q}\rightarrow[0,\infty]$ is
uniquely determined, monotone increasing, and continuous from below,
and where
$H_{0,Q}(X)$ is a canonical Orlicz premium given by the (generalized) Luxemburg norm
\begin{equation}
H_{0,Q} (X) \triangleq \inf \left \{ k >0 \; \Big|\; \mathbb{E}_{Q} \left [ \log \left ( \frac{X}{k} \right ) \right ] \leq 0 \right \}.
\end{equation}
Furthermore, quasi-logconvex star-shaped measures of risk are of the form \eqref{eq:repQLC} with $R$
geometrically expansive in the first coordinate,
i.e.,
\begin{equation}
R(s';Q)\geq R(s;Q)\exp\left(\vert \log s'-\log s\vert\right).
\label{eq:geoexpansive}
\end{equation}

We accomplish this by exploiting the discovery that \textit{cash-superadditivity} in arithmetic measurement of risk
corresponds to \textit{star-shapedness} in geometric measurement of risk.
Furthermore, \textit{quasi-convexity} in arithmetic risk measurement 
corresponds to \textit{quasi-logconvexity} in geometric risk measurement. 
We also reveal that when positive homogeneity for return risk measures
is replaced by star-shapedness,
then logconvexity should be replaced by quasi-logconvexity.

We motivate logconvex and quasi-logconvex measures of risk
as natural candidates when measuring the risk of continuously rebalanced investment portfolios,
as opposed to convex and quasi-convex risk measures that are more naturally applied to buy-and-hold portfolios.
We establish the dual representation \eqref{eq:repQLC} and several (sub)classification results
including an analysis of the precise intersection between quasi-convex and quasi-logconvex risk measures.
We also characterize quasi-logconvex risk measures in terms of properties of families of acceptance sets
and derive their law-invariant representation.
We provide several illustrative examples of quasi-logconvex risk measures
and analyze their implications in two important problems: 
portfolio choice and capital allocation.

The remainder of this paper is organized as follows.
In Section~\ref{sec:prel} we recall some preliminaries for monetary, return and related classes of risk measures.
In Section~\ref{sec:mot} we motivate logconvexity and quasi-logconvexity from a portfolio investment perspective.
In Section~\ref{sec:QLC} we introduce, characterize, and classify quasi-logconvex measures of risk.
Section~\ref{sec:acc} characterizes quasi-logconvex measures of risk in terms of properties of families of acceptance sets
and Section~\ref{sec:li} provides their law-invariant representation.
Finally, in Section~\ref{sec:examples} we provide examples
and in Section~\ref{sec:app} we describe applications of quasi-logconvex risk measures.

\setcounter{equation}{0}

\section{Preliminaries}\label{sec:prel}

We consider a measurable space $(\Omega, \F)$
and a reference probability measure $P$ defined on it.
We let $L^\infty(\Omega, \F, P)$, $L^\infty_+(\Omega, \F, P)$ and $L^\infty_{++}(\Omega, \F, P)$
be the sets of $P$-a.s.\,bounded, $P$-a.s.\,bounded non-negative, and $P$-a.s.\,bounded strictly positive random variables, respectively.
Equalities and inequalities between random variables hold $P$-a.s.
Throughout, we use the sign convention that positive (negative) realizations of
random variables represent losses (gains).

Furthermore, we denote by $\mathcal{Q}$ the set of all probability measures on this measurable space
that are absolutely continuous with respect to the reference probability measure $P$.
We also denote by $\mathcal{M}_{1,c} \triangleq \mathcal{M}_{1,c}(\Omega,\mathcal{F})$
the set of probability measures that have compact support in $\mathbb{R}$.


\subsection{Monetary, Return and Related Classes of Risk Measures}

A risk measure $\rho \colon L^{\infty}(\Omega, \F, P) \to \R$ is said to be
\begin{itemize}
\item [-]monotone if $X \leq Y \ P\mbox{-a.s. } \Rightarrow \rho(X) \leq \rho (Y)$;
\item [-]translation invariant if $\rho(X+h)=\rho(X)+h, \, \forall h \in \R, \forall X \in L^{\infty}$;
\item [-]positively homogeneous if $\rho( \lambda X)= \lambda \rho(X), \, \forall \lambda \geq 0, \forall X \in L^{\infty}$.
\end{itemize}

\begin{definition}[Monetary risk measure]
A monetary risk measure $\rho \colon L^{\infty} \to \R$ is a monotone and translation invariant risk measure
that is normalized to satisfy
$\rho (0) =0$.
\end{definition}

Bellini, Laeven and Rosazza Gianin \cite{BLR18} introduced return risk measures
on the smaller domain $L^\infty_{++}(\Omega, \F, P)$, as follows:
\begin{definition}[Return risk measure]
A return risk measure $\rho \colon L^{\infty}_{++} \to (0, \infty)$ is a monotone and positively homogeneous risk measure
that is normalized to satisfy
$\rho (1) =1$.
\end{definition}

Furthermore, a risk measure $\rho \colon L^{\infty}(\Omega, \F, P) \to \R$ is said to be
\begin{itemize}
\item [-]convex if it is monetary and
\[
\rho (\alpha X +(1-\alpha)Y) \leq \alpha \rho (X) + (1 - \alpha) \rho (Y), \, \forall X,Y \in L^{\infty}, \forall \alpha \in [0,1];
\]
\item [-]coherent if it is convex and positively homogeneous;
\item [-]subadditive if
\[
\rho (X + Y) \leq \rho (X) + \rho (Y), \, \forall X,Y \in L^{\infty};
\]
\item [-]cash-superadditive (cash-subadditive) if
\[
\rho(X+h)\geq (\leq)\ \rho(X)+h, \, \forall h \in \R_{+}, \forall X \in L^{\infty};
\]
\item [-]quasi-convex if it is monotone and
\[
\rho (\alpha X +(1-\alpha)Y) \leq \max\{\rho (X), \rho (Y)\}, \, \forall X,Y \in L^{\infty}, \forall \alpha \in (0,1).
\]
\item [-]continuous from below if
\[
X_n \uparrow X \Rightarrow \rho(X_n) \to \rho (X).
\]
\item [-]continuous from above if
\[
X_n \downarrow X \Rightarrow \rho(X_n) \to \rho (X).
\]
\end{itemize}

Note that $\rho(X+h)\geq (\leq)\ \rho(X)+h, \, \forall h \in \R_{+}$ is equivalent to
$\rho(X-h)\leq (\geq)\ \rho(X)-h, \, \forall h \in \R_{+}$. Indeed, if $\rho(X+h)\geq (\leq)\ \rho(X)+h$ holds for any $X \in L^{\infty}$ and $h \in \R_{+}$, then $\rho(X)=\rho(X+h-h) \geq (\leq)\ \rho(X-h)+h$. Hence $\rho(X-h)\leq (\geq)\ \rho(X)-h$ for any $X \in L^{\infty}$ and $h \in \R_{+}$. The converse implication can be proved similarly.\bigskip

We state the following lemma (cf. F\"ollmer and Schied \cite{FS11}, Chapter 4).
\begin{lemma}
\begin{itemize}
\item[(i)] Suppose $\rho$ is monotone, translation invariant, and positively homogeneous with $\rho(0)=0$.
That is, $\rho$ is both a monetary and a return risk measure.
Then $\rho$ is convex if and only if it is subadditive.
\item[(ii)] Suppose $\rho$ is monotone and translation invariant with $\rho(0)=0$, i.e., $\rho$ is monetary.
Then $\rho$ is convex if and only if it is quasi-convex.
\end{itemize}
\label{lem:1a}
\end{lemma}

In Figure \ref{fig:I} we illustrate the connections between monetary, return and related classes of risk measures.
Clearly, the class of coherent risk measures is a (strict) subclass of the intersection between monetary and return risk measures,
as illustrated in panel (c).
From (i) in Lemma \ref{lem:1a} we conclude that the only intersection between convex and return risk measures
is given by the class of coherent risk measures; see panel (d).
From Theorems 6.1 and 6.2 in Laeven and Stadje \cite{LS13} we obtain the decomposition of convex risk measures
into entropy convex, entropy coherent, and coherent measures of risk
illustrated in panels (e) and (f).
From (ii) in Lemma \ref{lem:1a} we conclude that the only intersection between quasi-convex and monetary risk measures
is given by the class of convex risk measures; see panel (g).
In panel (h) we illustrate the (strict) subclass of quasi-convex risk measures that are cash-subadditive.

\begin{figure}
\centering

\subfigure[Monetary]
{

\begin{tikzpicture}[scale = 0.34]
\draw[red,thick,rotate around={45:(0,0)}] (2,0) rectangle ++(6,6);
\end{tikzpicture}

}
\ \ \ \ \ \ \ \ \ \
\subfigure[Monetary and Return]
{

\begin{tikzpicture}[scale = 0.34]
\draw[blue,thick,fill=blue!10!white,rotate around={45:(6,0)}] (6,0) rectangle ++(6,6);
\draw[red,thick,rotate around={45:(2,0)}] (2,0) rectangle ++(6,6);
\end{tikzpicture}

}
\subfigure[Monetary, Return and Coherent]
{

\begin{tikzpicture}[scale = 0.34]
\draw[blue,thick,fill=blue!10!white,rotate around={45:(6,0)}] (6,0) rectangle ++(6,6);
\draw[red,thick,rotate around={45:(2,0)}] (2,0) rectangle ++(6,6);
\draw[violet,thick,rotate around={45:(2.53,3.53)}] (2.53,3.53) rectangle ++(1,1);
\end{tikzpicture}

}
\ \ \ \ \
\subfigure[Monetary, Return, Coherent and Convex]
{

\begin{tikzpicture}[scale = 0.34]
\draw[blue,thick,fill=blue!10!white,rotate around={45:(6,0)}] (6,0) rectangle ++(6,6);
\draw[red,thick,rotate around={45:(2,0)}] (2,0) rectangle ++(6,6);
\draw[violet,thick,rotate around={45:(2.53,3.53)}] (2.53,3.53) rectangle ++(1,1);
\draw[green,thick,rotate around={45:(0.5,1.4)}] (0.5,1.4) rectangle ++(4,4);
\end{tikzpicture}

}
\subfigure[Monetary, Return, Coherent, Convex and Entropy Convex]
{

\begin{tikzpicture}[scale = 0.34]
\draw[blue,thick,fill=blue!10!white,rotate around={45:(6,0)}] (6,0) rectangle ++(6,6);
\draw[red,thick,rotate around={45:(2,0)}] (2,0) rectangle ++(6,6);
\draw[violet,thick,rotate around={45:(2.53,3.53)}] (2.53,3.53) rectangle ++(1,1);
\draw[green,thick,rotate around={45:(0.5,1.4)}] (0.5,1.4) rectangle ++(4,4);
\draw[purple,thick,rotate around={45:(-0.25,2.25)}] (-0.25,2.25) rectangle ++(4,2.8);
\end{tikzpicture}

}
\ \ \ \ \
\subfigure[Monetary, Return, Coherent, Convex, Entropy Convex and Entropy Coherent]
{

\begin{tikzpicture}[scale = 0.34]
\draw[blue,thick,fill=blue!10!white,rotate around={45:(6,0)}] (6,0) rectangle ++(6,6);
\draw[red,thick,rotate around={45:(2,0)}] (2,0) rectangle ++(6,6);
\draw[violet,thick,rotate around={45:(2.53,3.53)}] (2.53,3.53) rectangle ++(1,1);
\draw[green,thick,rotate around={45:(0.5,1.4)}] (0.5,1.4) rectangle ++(4,4);
\draw[purple,thick,rotate around={45:(-0.25,2.25)}] (-0.25,2.25) rectangle ++(4,2.8);
\draw[yellow,thick,rotate around={45:(1.8,4.2)}] (1.8,4.2) rectangle ++(1.2,2.8);
\end{tikzpicture}

}
\subfigure[Monetary, 
Coherent, Convex, Entropy Convex, Entropy Coherent and Quasi-Convex]
{

\begin{tikzpicture}[scale = 0.34]
\draw[red,thick,rotate around={45:(2,0)}] (2,0) rectangle ++(6,6);
\draw[violet,thick,rotate around={45:(2.53,3.53)}] (2.53,3.53) rectangle ++(1,1);
\draw[green,thick,rotate around={45:(0.5,1.4)}] (0.5,1.4) rectangle ++(4,4);
\draw[purple,thick,rotate around={45:(-0.25,2.25)}] (-0.25,2.25) rectangle ++(4,2.8);
\draw[yellow,thick,rotate around={45:(1.8,4.2)}] (1.8,4.2) rectangle ++(1.2,2.8);
\draw[pink,thick,rotate around={45:(-1,0)}] (-1,0) rectangle ++(6,6);
\end{tikzpicture}

}
\ \ \
\subfigure[Monetary, 
Coherent, Convex, Entropy Convex, Entropy Coherent, Quasi-Convex and Cash-Subadditive]
{

\begin{tikzpicture}[scale = 0.34]
\draw[red,thick,rotate around={45:(2,0)}] (2,0) rectangle ++(6,6);
\draw[violet,thick,rotate around={45:(2.53,3.53)}] (2.53,3.53) rectangle ++(1,1);
\draw[green,thick,rotate around={45:(0.5,1.4)}] (0.5,1.4) rectangle ++(4,4);
\draw[purple,thick,rotate around={45:(-0.25,2.25)}] (-0.25,2.25) rectangle ++(4,2.8);
\draw[yellow,thick,rotate around={45:(1.8,4.2)}] (1.8,4.2) rectangle ++(1.2,2.8);
\draw[pink,thick,rotate around={45:(-1,0)}] (-1,0) rectangle ++(6,6);
\draw[magenta,thick,rotate around={45:(-0.25,0.62)}] (-0.25,0.62) rectangle ++(5,5);
\end{tikzpicture}

}
\caption{Monetary, Return and Related Measures of Risk}
\small This figure illustrates the connections between monetary, return and related classes of risk measures.
(We restrict to $L^{\infty}_{++}$.)
In panels~(g) and~(h) we omit return risk measures; the precise intersection between quasi-convex and return risk measures is analyzed in Section~\ref{sec:QLC}.
\label{fig:I}
\end{figure}

\subsection{One-to-One Correspondence}

Recall the one-to-one correspondence between monetary risk measures and return risk measures (Bellini, Laeven and Rosazza Gianin \cite{BLR18}):
\begin{itemize}
\item[-] given a monetary risk measure $\rho \colon  L^{\infty} \to \R $,
the associated return risk measure $\trho \colon L^{\infty}_{++} \to (0, \infty)$ is given by
\begin{equation}
\label{eq:trho-1}
\boxed{
\trho(X)\triangleq\exp \left( \rho \left( \log(X) \right) \right);
}
\end{equation}
\item[-] given a return risk measure $\trho \colon L^{\infty}_{++} \to (0, \infty)$,
the associated monetary risk measure $\rho \colon  L^{\infty} \to \R $ is given by
\begin{equation}
\label{eq:trho-2}
\boxed{
\rho(Z)\triangleq\log \left( \trho \left( \exp(Z) \right) \right).
}
\end{equation}
\end{itemize}

For a risk measure $\trho \colon L^{\infty}_{++} \to (0,\infty)$, we say that $\trho$ is:
\begin{itemize}
\item [-]constant-multiplicative if $\trho(X^\alpha)=\trho^\alpha(X), \; \forall X \in L^{\infty}_{++}$, $\alpha>0$;
\item [-]submultiplicative if $\trho(XY) \leq \trho(X) \trho(Y), \; \forall X, Y \in L^{\infty}_{++}$;
\item [-]logconvex if it is monotone, positively homogeneous and
$\trho(X^\alpha Y^{1-\alpha}) \leq \trho^\alpha(X)\trho^{1-\alpha}(Y), \; \forall X,Y \in L^{\infty}_{++}$, $\alpha \in (0,1)$.
\end{itemize}

The following lemma follows straightforwardly:
\begin{lemma}
For $\rho \colon  L^{\infty} \to \R $ and $\trho \colon L^{\infty}_{++} \to (0, \infty)$ as defined in \eqref{eq:trho-1} and \eqref{eq:trho-2},
we have the following equivalences:
\begin{itemize}
\item [(a)] $\rho(0)=0 \iff \trho(1)=1$
\item [(b)] $\rho$ is monotone $\iff$ $\trho$ is monotone
\item [(c)] $\rho$ is translation invariant $\iff$ $\trho$ is positively homogeneous
\item [(d)] $\rho$ is positively homogeneous $\iff$ $\trho$ is constant-multiplicative
\item [(e)] $\rho$ is subadditive $\iff$ $\trho$ is submultiplicative
\item [(f)] $\rho$ is convex $\iff$ $\trho$ is logconvex
\end{itemize}
\label{lem:rho-trho}
\end{lemma}

\begin{remark}
Logconvexity may also be referred to as geometric convexity or multiplicative convexity.
Geometrically convex functions have received some interest in the theory of convex functions,
in particular in connection with the Hermite-Hadamard inequality.
\end{remark}

\subsection{Examples of Return Risk Measures}

We provide some examples of return risk measures.

\begin{example}[V@R] \label{ex: Var}
Consider the Value-at-Risk (V@R), or quantile, at probability level $\alpha\in(0,1)$ defined as\footnote{Note that the following  definition of quantile is usually known as the smallest $\alpha$-quantile of $X$. It holds also that $q_{ \alpha}\left(X\right)= q_{1-\alpha} ^+ \left( -X \right)$, where $q_{\beta} ^+ \left(Z\right)\triangleq \inf \{ x\in\mathbb{R} : P[Z\leq x]>\beta \}$ is the largest $\beta$-quantile. See F\"{o}llmer and Schied \cite{FS11}.}
\begin{align}
\rho\left(X\right)
=\inf \{ x\in\mathbb{R} : P[X\leq x]\geq\alpha \}
\triangleq F_{X}^{-1}(\alpha)\triangleq q_{\alpha}\left(X\right).
\label{eq:V@R}
\end{align}
As is well-known, V@R is monotone, translation invariant, and positively homogeneous with $\rho(0)=0$.
Hence, it is both a monetary and a return risk measure.
In fact, under V@R, risk assessment at the level of returns coincides with risk assessment at the level of monetary values
(cf. Example 1 of Bellini, Laeven and Rosazza Gianin \cite{BLR18}).
V@R is, however, not subadditive, hence not coherent.
\end{example}


The following example introduces the Average-Return-at-Risk.
\begin{example}[AR@R] \label{ex: ARaR}
The Average-Return-at-Risk (AR@R)
at probability level $\alpha\in(0,1)$ is defined as
the return counterpart of the Average-Value-at-Risk (AV@R), i.e.,
\begin{eqnarray*}
\trho\left(X\right)&=& \exp\left( \rho(\log (X) \right) 
=\exp\left( \int_{\alpha}^{1}\frac{1}{1-\alpha}F_{\log X}^{-1}(\beta)\,\mathrm{d}\beta\right)\\
&=&\exp\left( \int_{\alpha}^{1}\frac{1}{1-\alpha}\log F_{X}^{-1}(\beta)\,\mathrm{d}\beta\right) 
=\left(\exp\left( \int_{\alpha}^{1}\log F_{X}^{-1}(\beta)\,\mathrm{d}\beta\right)\right)^{\frac{1}{1-\alpha}},
\label{eq:AR@R}
\end{eqnarray*}
where the far right-hand side can be given the interpretation of a geometric average.
AR@R is monotone, positively homogeneous, and submultiplicative with $\trho(1)=1$, but not translation invariant.
\end{example}

\begin{example}[Orlicz Premium and HG-Risk Measure]
Consider the extension to the Rockafellar-Uryasev construction (\cite{RU00},\cite{RUZ08}) given by the Haezendonck-Goovaerts (HG) risk measure
defined as
\begin{align}
\rho\left(X\right)
=\inf_{x\in\mathbb{R}}\left\{x+H_{\Phi,\alpha}\left[\left(X-x\right)^{+}\right]\right\},
\label{eq:HG}
\end{align}
where
\begin{equation}
H_{\Phi,\alpha}\left[X\right]=\inf\left\{k>0:\mathbb{E}\left[\Phi\left(\frac{X}{k}\right)\right]\leq 1-\alpha\right\},
\label{eq:Orlicz}
\end{equation}
with $\Phi:[0,\infty)\rightarrow[0,\infty)$, $\Phi(0)<1<\Phi(\infty)$ and non-decreasing,
is the Orlicz premium.
The Orlicz premium is monotone and positively homogeneous. It is also normalized (that is $H_{\Phi,\alpha}\left[1\right]=1$) iff $\Phi(1)=1-\alpha$. Hence it is a return risk measure if $\Phi(1)=1-\alpha$.
It is the return counterpart of (utility-based) shortfall risk.
The HG-risk measure is monotone, translation invariant, and positively homogeneous. It is also normalized (that is $\rho(0)=0$) if $\Phi(1)=1-\alpha$.
Thus, it is both a monetary and a return risk measure if $\Phi(1)=1-\alpha$.
When $\Phi$ is convex, it is moreover coherent.
(See Bellini, Laeven and Rosazza Gianin \cite{BLR18}) for corresponding dual representations.)
\end{example}

%

\begin{example}[Logconvex Risk Measures]
Consider the return counterpart of the monetary convex risk measures given by:
\begin{eqnarray*}
\trho(X) &=& \exp\left( \rho(\log (X) \right) 
=\exp\left( \sup_{Q \in \mathcal{Q}} \left\{\mathbb{E}_{Q} [\log(X) ]- c(Q)\right\} \right) \\
&=& \sup _{Q \in \mathcal{Q}} \beta(Q) \exp \left( \mathbb{E}_{Q} [\log(X) ] \right),
\end{eqnarray*}
with $\beta(Q)=\exp(-c(Q))$.
Logcoherent risk measures occur when
\begin{align*}
c(Q)=\left\{
       \begin{array}{ll}
         0, & \mathrm{if}\ Q\in M; \\
         \infty, & \mathrm{otherwise}.
       \end{array}
     \right.
\end{align*}
\end{example}

\begin{example}[$p$-norm]
Consider the $p$-norm
\begin{equation}
\rho\left(X\right)=\mathbb{E}\left[\left(X\right)^{\gamma}\right]^{1/\gamma}\triangleq\Vert X \Vert_\gamma,\qquad \gamma>0.
\label{eq:pnorm}
\end{equation}
It occurs as the certainty equivalent under expected utility with constant relative risk aversion (CRRA), or isoelastic, utility.
It is the return counterpart of the monetary entropic risk measure.
\end{example}

\begin{example}[Robust $p$-norm]
Consider the return counterpart of the monetary entropy coherent risk measure given by:
\begin{eqnarray*}
\trho(X) &=& \exp\left( \rho(\log (X) \right) 
=\exp\left( \sup_{Q \in M\subset\mathcal{Q}} \frac{1}{\gamma}\log\left(\mathbb{E}_{Q} [\exp(\gamma\log(X))]\right)\right) \\
&=& \sup _{Q \in M\subset \mathcal{Q}} \left( \mathbb{E}_{Q} [X^{\gamma} ] \right)^{1/\gamma},\qquad \gamma>0.
\end{eqnarray*}
\end{example}

\begin{example}[Robust discounted $p$-norm]
Consider the return counterpart of the monetary entropy convex risk measure given by:
\begin{eqnarray*}
\trho(X) &=& \exp\left( \rho(\log (X)) \right) 
=\exp\left( \sup_{Q \in \mathcal{Q}} \left\{\frac{1}{\gamma}\log\left(\mathbb{E}_{Q} [\exp(\gamma\log(X))]\right)-c(Q)\right\}\right) \\
&=& \sup _{Q \in \mathcal{Q}} \beta(Q) \left( \mathbb{E}_{Q} [X^{\gamma} ] \right)^{1/\gamma},\qquad \gamma>0,
\end{eqnarray*}
with $\beta(Q)=\exp(-c(Q))$, as before.
\end{example}

\setcounter{equation}{0}

\section{(Quasi-)Convexity vs.~(Quasi-)Logconvexity: Buy-and-Hold vs.~Rebalancing Strategies}\label{sec:mot}

The key ``take away'' from the following exposition is that convex combinations of \textit{returns} occur naturally when rebalancing portfolios,
whereas convex combinations of monetary \textit{values} occur naturally when executing buy-and-hold strategies.
As rebalancing may be a more typical investment strategy than a buy-and-hold strategy,
requiring (quasi-)logconvexity appears to be more natural than requiring (quasi-)convexity, when evaluating portfolio risk.

Imagine an investor who invests her initial wealth $W=W_{0}$ in two assets, $A$ and $B$.
Denote by $V_{A}=V_{A,0}$ and $V_{B}=V_{B,0}$ the (absolute, dollar) amounts invested in assets $A$ and $B$, respectively.
Clearly, the total amount invested satisfies $W=V_{A}+V_{B}$.
The portfolio weights, or shares of dollars, invested in $A$ and $B$ are given by
$w_{A}=\frac{V_{A}}{W}$ and
$w_{B}=\frac{V_{B}}{W}$.
Of course, $w_{A}+w_{B}=1$.
Furthermore, denote by $\tilde{V}_{A,1}$ and $\tilde{V}_{B,1}$
the wealths at the end of period 1 that can be generated by fully investing the initial wealth $W$ 
in assets $A$ and $B$, with $\tilde{V}_{A,0}=W$ and $\tilde{V}_{B,0}=W$.

Let $R_{A}$ and $R_{B}$ be the one-period returns on $A$ and $B$.
Then, the value of the invested wealth, or portfolio, at the end of the first period is given by
\begin{equation*}
W_{1}=W\left(w_{A}(1+R_{A})+w_{B}(1+R_{B})\right)=W(1+R),
\end{equation*}
where the portfolio return
\begin{equation*}
R=w_{A}R_{A}+w_{B}R_{B},
\end{equation*}
because $w_{A}+w_{B}=1$.
That is, the portfolio return is given by the weighted average of the returns on $A$ and $B$.
Note also that $W_{1}=w_{A} \tilde{V}_{A,1}+w_{B} \tilde{V}_{B,1}=w_{A} \tilde{V}_{A,1}+(1-w_{A}) \tilde{V}_{B,1}$.

It is common practice, in a multi-period setting,
to let portfolio weights remain constant over time.
That is, the investor \textit{rebalances} the portfolio at the end of each time period (which may be infinitesimally small),
such that constant portfolio weights are maintained.
Alternatively, the investor may execute a \textit{buy-and-hold} strategy,
where no rebalancing takes place.

Under no rebalancing, the value of the portfolio
at the end of period $t$
is simply given by
$W_{t}=V_{A,t}+V_{B,t}$.
Note also that $W_{t}=w_{A} \tilde{V}_{A,t}+w_{B} \tilde{V}_{B,t}=w_{A} \tilde{V}_{A,t}+(1-w_{A}) \tilde{V}_{B,t}$.

Convexity, in the regular sense, requires that
\begin{align}
\rho(W_{t})&=\rho(w_{A} \tilde{V}_{A,t}+(1-w_{A}) \tilde{V}_{B,t})\nonumber\\
&\leq w_{A} \rho(\tilde{V}_{A,t})+(1-w_{A}) \rho(\tilde{V}_{B,t}),
\label{eq:buy-hold-convex}
\end{align}
to capture ``diversification'' at the level of absolute, monetary amounts.
Quasi-convexity requires that
\begin{align}
\rho(W_{t})&=\rho(w_{A} \tilde{V}_{A,t}+(1-w_{A}) \tilde{V}_{B,t})\nonumber\\
&\leq \max\{\rho(\tilde{V}_{A,t}),\rho(\tilde{V}_{B,t})\}.
\label{eq:buy-hold-QC}
\end{align}
Thus, convexity and quasi-convexity are intimately related to buy-and-hold portfolios.

Under the more common rebalancing of the portfolio,
the return in every time period is given by
\begin{equation*}
R_{t}=w_{A} R_{A,t}+w_{B} R_{B,t},
\end{equation*}
such that
\begin{equation*}
W_{t}=W\prod_{j=1}^{t}(1+R_{j}),
\end{equation*}
with
\begin{equation*}
R_{j}=w_{A}R_{A,j}+w_{B}R_{B,j}.
\end{equation*}

Under continuous compounding and continuous rebalancing, 
\begin{equation*}
W_{t}=W\exp\left(w_{A} r_{A,t}+w_{B} r_{B,t}\right).
\end{equation*}
Indeed, starting from
\begin{equation*}
W_{t}=W\prod_{j=1}^{t}(1+R_{j}),
\end{equation*}
the wealth at time $t+h$ (with $h>0$) can be obtained recursively as
\begin{equation*}
W_{t+h}=W_t (1+R_{t,t+h}),
\end{equation*}
with $R_{t,t+h}$ denoting the return on the period from $t$ to $t+h$. By this one-step relation and by using arguments on continuous compounding, 
\begin{equation*}
\frac{W_{t+h}-W_t}{W_t}=R_{t,t+h}.
\end{equation*}
Dividing both the terms by $h$, taking the limit as $h\to 0$, and setting $\bar{r}_t \triangleq \lim_{h \to 0} \frac{R_{t,t+h}}{h}$,
\begin{eqnarray*}
&&\lim_{h \to 0} \frac{R_{t,t+h}}{h}= \lim_{h \to 0} \frac{W_{t+h}-W_t}{h W_t}\\
&&\bar{r}_t= \frac{1}{W_t} \frac{dW_t}{dt}.
\end{eqnarray*}
It then follows that
\begin{equation*}
W_t=W  \exp \left(\int_0^t \bar{r}_s \, ds\right)=W e^{r_t},
\end{equation*}
with $r_t \triangleq \int_0^t \bar{r}_s \, ds$.
Using similar arguments as above, the relation $R_{t, t+h}=w_{A} R_{A,t, t+h}+w_{B} R_{B,t, t+h}$ then implies that
\begin{equation*}
r_t= w_{A} r_{A,t}+w_{B} r_{B,t},
\end{equation*}
hence
\begin{equation*}
W_{t}=W\exp\left(w_{A} r_{A,t}+w_{B} r_{B,t}\right).
\end{equation*}

Logconvexity, or ``geometric convexity'', requires that
\begin{align*}
\tilde{\rho}(W_{t})&=\tilde{\rho}\left(W\exp\left(w_{A} r_{A,t}+w_{B} r_{B,t}\right)\right)\\
&=W\tilde{\rho}\left(\left(\exp\left(r_{A,t}\right)\right)^{w_{A}}\left(\exp\left(r_{B,t}\right)\right)^{1-w_{A}}\right)\\
&=W\exp\left(\rho\left(w_{A} r_{A,t}+(1-w_{A}) r_{B,t}\right)\right)\\
&\leq W\exp\left(w_{A} \rho(r_{A,t})+(1-w_{A}) \rho(r_{B,t})\right)\\
&=W\left(\tilde{\rho}(\exp(r_{A,t}))\right)^{w_{A}}\left(\tilde{\rho}(\exp(r_{B,t}))\right)^{(1-w_{A})},
\end{align*}
to capture ``diversification'' at the level of relative amounts, or returns.

In short, under logconvexity,
\begin{align}
\tilde{\rho}(W_{t})&=W\tilde{\rho}\left(\left(\exp\left(r_{A,t}\right)\right)^{w_{A}}\left(\exp\left(r_{B,t}\right)\right)^{1-w_{A}}\right)\nonumber\\
&\leq W\tilde{\rho}^{w_{A}}(\exp(r_{A,t}))\tilde{\rho}^{(1-w_{A})}(\exp(r_{B,t})).
\label{eq:rebalancing-logconvex}
\end{align}
Quasi-logconvexity requires that
\begin{align}
\tilde{\rho}(W_{t})&=\tilde{\rho}\left(W\left(\exp\left(r_{A,t}\right)\right)^{w_{A}}\left(\exp\left(r_{B,t}\right)\right)^{1-w_{A}}\right)\nonumber\\
&=\tilde{\rho}\left(\left(W\exp\left(r_{A,t}\right)\right)^{w_{A}}\left(W\exp\left(r_{B,t}\right)\right)^{1-w_{A}}\right)\nonumber\\
&\leq \max\{\tilde{\rho}(W\exp(r_{A,t})),\tilde{\rho}(W\exp(r_{B,t}))\}.
\label{eq:rebalancing-QLC}
\end{align}
Thus, logconvexity and quasi-logconvexity are intimately related to rebalanced portfolios.
Cf.~\eqref{eq:buy-hold-convex}--\eqref{eq:buy-hold-QC} to~\eqref{eq:rebalancing-logconvex}--\eqref{eq:rebalancing-QLC}.

%
%
%
%

\setcounter{equation}{0}

\section{Diversification and Liquidity}\label{sec:QLC}


\subsection{On Quasi-Logconvex and Star-Shaped Risk Measures}\label{sec: qlc and star-shaped}

By serendipity, the one-to-one correspondence between monetary and return risk measures
also gives rise to a new class of risk measures to stand on equal footing with the class of quasi-convex risk measures.

For a mapping $\rho \colon L^{\infty} \to [-\infty,\infty]$,
let the associated $\trho \colon L^{\infty}_{++} \to [0, \infty]$
be given by
\begin{equation}
\label{eq:trho-3}
\trho(X)\triangleq\exp \left( \rho \left( \log(X) \right) \right),
\end{equation}
and for a mapping $\trho \colon L^{\infty}_{++} \to [0, \infty]$,
let the associated $\rho \colon  L^{\infty} \to [-\infty,\infty]$ be given by
\begin{equation}
\label{eq:trho-4}
\rho(Z)\triangleq\log \left( \trho \left( \exp(Z) \right) \right);
\end{equation}
cf.~\eqref{eq:trho-1}--\eqref{eq:trho-2} but without restricting to monetary risk measures $\rho$ and return risk measures $\trho$, and extending their ranges.
By analogy to geometric and arithmetic means, we will refer to~\eqref{eq:trho-3} and~\eqref{eq:trho-4} as geometric and arithmetic counterparts, respectively.
Note that, differently from before, in the following, $\rho$ may assume the values $-\infty$ and $+\infty$ and $\trho$ the values $0$ and $+\infty$. 
This choice is motivated by the fact that (quasi-convex) risk measures defined on $L^{\infty}$ are not necessarily finite, in general. 
Finiteness is guaranteed, e.g., under normalization, monotonicity, and cash-additivity.
\medskip

We say that $\trho \colon L^{\infty}_{++} \to [0, \infty]$ is:
\begin{itemize}
\item [(i)]
star-shaped if $\trho( \lambda X)\geq \lambda \trho(X), \, \forall \lambda \geq 1, \forall X \in L^{\infty}_{++}$.
\item [(ii)]
quasi-logconvex if $\trho(X^\alpha Y^{1-\alpha}) \leq \max\{\trho(X),\trho(Y)\}, \; \forall X,Y \in L^{\infty}_{++}$, $\forall \alpha \in (0,1)$.
\end{itemize}

The following holds:
\begin{lemma}
\label{lem: rho_trho-2}
Let $\rho \colon  L^{\infty} \to [-\infty,\infty]$ and $\trho \colon L^{\infty}_{++} \to [0, \infty]$ be as in \eqref{eq:trho-3} and \eqref{eq:trho-4}.
Then:
\begin{itemize}
\item [(0)] The equivalences in Lemma \ref{lem:rho-trho} extend to $\rho$ and $\trho$ as in \eqref{eq:trho-3}--\eqref{eq:trho-4}.
\item [(a)] $\rho$ is cash-superadditive $\iff$ $\trho$ is star-shaped
\item [(b)] $\rho$ is quasi-convex $\iff$ $\trho$ is quasi-logconvex
\item [(c)] $\rho$ is continuous from above (resp. from below) $\iff$ $\trho$ is continuous from above away from $0$ (resp. from below), where continuity from above away from $0$ means that if $X_n \downarrow X$ with $X_n,X \in L^{\infty}_{++}$ then $\lim_{n \to + \infty} \trho(X_n)=\trho(X)$.
\end{itemize}
\end{lemma}

\begin{proof}
(0) follows immediately by Lemma~\ref{lem:rho-trho}.

(a) Assume that $\rho$ is cash-superadditive. 
It then follows that for any $X \in L^{\infty}_{++}$ and $\lambda \geq 1$,
\begin{eqnarray*}
\trho(\lambda X) &=& \exp \left( \rho \left( \log(\lambda X) \right) \right) \\
 &=& \exp \left( \rho \left( \log(X) + \log \lambda  \right) \right) \\
  &\geq & \exp \left( \rho \left( \log(X) \right)+ \log \lambda \right) \\
  &=& \lambda \trho (X).
\end{eqnarray*}

Viceversa: by star-shapedness of $\trho$ it follows that for any $X \in L^{\infty}$ and $h \in \Bbb R_+$,
\begin{eqnarray*}
\rho(X+h) &=& \log \left( \trho \left( e^{X+h} \right) \right) \\
 &\geq & \log \left(e^h \trho \left( e^{X} \right) \right) \\
  &= & \log \left( \trho \left( e^{X} \right) \right) +h\\
  &=& \rho (X)+h.
\end{eqnarray*}

(b) Note that $\trho$ is quasi-logconvex $\iff \{\log X \in L^{\infty}|\rho(\log X)\leq \alpha\}
=\{\log X \in L^{\infty}|\trho(X)\leq \exp(\alpha)\}$ $\forall \alpha \in \R$ is convex.
In other words, convex lower level sets of absolute positions under 
risk measures of the form \eqref{eq:trho-3}
corresponds to convex lower level sets of relative (log) positions under risk measures of the form \eqref{eq:trho-4}.

%
%

(c) We prove the case of continuity from above. The case of continuity from below can be checked similarly.

Let $\rho$ be continuous from above and let $(X_n)_{ n \in \Bbb N}$ be a sequence in $L^{\infty}_{++}$ with $X_n \downarrow X$ and $X \in L^{\infty}_{++}$.
Then,
\begin{equation*}
\lim_{n \to + \infty} \trho(X_n)=\exp\left\{\lim_{n \to + \infty}  \rho \left( \log( X_n) \right) \right\}=\trho(X),
\end{equation*}
by increasing monotonicity of $\log(\cdot)$ and by continuity from above of $\rho$.

Viceversa, let $\trho$ be continuous from above away from $0$ and let $(X_n)_{ n \in \Bbb N}$ be a sequence in $L^{\infty}$.
Then,
\begin{equation*}
\lim_{n \to + \infty} \rho(X_n)=\log\left\{\lim_{n \to + \infty}  \trho \left( \exp( X_n) \right) \right\}=\rho(X),
\end{equation*}
because $\exp(X_n) \downarrow \exp(X) \in L^{\infty}_{++}$.

\end{proof}

Also note that
$\trho( \lambda X)\geq \lambda \trho(X), \, \forall \lambda \geq 1, \forall X \in L^{\infty}_{++}$
is equivalent to
$\trho( \lambda X)\leq \lambda \trho(X), \, \forall 0<\lambda \leq 1, \forall X \in L^{\infty}_{++}$.

We recall that any convex 
risk measure is also star-shaped. 
See Frittelli and Rosazza Gianin \cite{FR02}, Remark~8, for a proof. 
See also the very recent \cite{CCMTW22} for an analysis of monetary star-shaped risk measures.

\subsection{Representation Results}




Thanks to the dual representation of quasi-convex risk measures
(see Cerreia-Vioglio et al. \cite{CMMM11}, Drapeau and Kupper \cite{DK13}, and Frittelli and Maggis \cite{FM11}),
to the one-to-one correspondences \eqref{eq:trho-3} and \eqref{eq:trho-4},
and to Lemmas \ref{lem:rho-trho} and \ref{lem: rho_trho-2},
we are now able to provide the dual representations of monotone, quasi-logconvex and star-shaped risk measures:

\begin{theorem} \label{thm: dual repres rho-tilde}
\textbf{[Monotonicity and quasi-logconvexity.
$\Leftrightarrow$
Representation with $\mathfrak{R}$.]}
\smallskip

The mapping $\trho \colon L^{\infty}_{++} \to [0, \infty]$ is monotone, quasi-logconvex and continuous from below
if and only if it has the following dual representation:
\begin{equation} \label{eq: dual repres rho-tilde}
\trho(X)= \sup_{Q \in \mathcal{Q}} \{ \exp \left(\mathfrak{R}(\mathbb{E}_{Q} [\log X];Q) \right) \},
\end{equation}
where $\mathcal{Q}$ is a set of probability measures that are absolutely continuous with respect to $P$,
and $\mathfrak{R}:\Bbb R \times \mathcal{Q} \to [-\infty,\infty]$
is a functional $\mathfrak{R}(t;Q)$ that is monotone increasing and continuous from below in $t$ for any $Q$.

Moreover, given $\trho$ as above, $\mathfrak{R}$ is defined by
\begin{equation} \label{eq: definition-R}
\mathfrak{R}(t;Q) \triangleq \inf_{Y \in L^{\infty}_{++}} \{\log (\trho(Y)): \mathbb{E}_Q [\log Y] \geq t \}, \quad t \in \Bbb R,\quad Q \in \mathcal{Q}.
\end{equation}
\end{theorem}

\begin{proof}
From Lemmas~\ref{lem:rho-trho} and~\ref{lem: rho_trho-2}, it follows that $\rho(X)=\log(\trho(e^X))$
is a monotone, quasi-convex and continuous from below risk measure.
Hence, by Prop. 2.13 of Frittelli and Maggis \cite{FM11},  Cerreia-Vioglio et al. \cite{CMMM11}, Theorem~3.1 and Lemma~3.2,
and Drapeau and Kupper \cite{DK13}, $\rho$ has the following dual representation:
\begin{equation} \label{eq: dual repres rho}
\rho(X)=\sup_{Q\in\mathcal{Q}} \left(\mathfrak{R}\left(\mathbb{E}_{Q}\left[ X\right]; Q\right)\right),
\end{equation}
where $\mathfrak{R}:\mathbb{R}\times\mathcal{Q}\rightarrow[-\infty,\infty]$ is
monotone increasing and continuous from below in $t \in \Bbb R$.

By \eqref{eq:trho-3}--\eqref{eq:trho-4} and \eqref{eq: dual repres rho} we deduce that
\begin{eqnarray*}
\trho(X) &=& \exp\left( \rho(\log (X) \right) \\
&=& \exp \left\{ \sup_{Q\in\mathcal{Q}} \left(\mathfrak{R}\left(\mathbb{E}_{Q}\left[\log X\right]; Q\right)\right) \right\} \\
&=& \sup_{Q \in \mathcal{Q}} \{ \exp \left(  \mathfrak{R}(\mathbb{E}_{Q} [\log X];Q) \right) \}.
\end{eqnarray*}

Furthermore, from Frittelli and Maggis \cite{FM11} we know that
$\mathfrak{R}(t;Q) = \inf_{X \in L^{\infty}} \{\rho(X): \mathbb{E}_Q [X] \geq t \}$ for any $t \in \Bbb R, Q \in \mathcal{Q}$.
Hence,
\begin{eqnarray*}
\mathfrak{R}(t;Q) &=& \inf_{X \in L^{\infty}} \{\rho(X): \mathbb{E}_Q [X] \geq t \} \\
&=& \inf_{X \in L^{\infty}} \{\log (\trho(e^X)): \mathbb{E}_Q [X] \geq t \} \\
&=& \inf_{Y \in L^{\infty}_{++}} \{\log (\trho(Y)): \mathbb{E}_Q [\log Y] \geq t \}.
\end{eqnarray*}
\end{proof}

Setting $\tilde{\mathfrak{R}}(t;Q)= \exp \mathfrak{R}(t;Q)$, the dual representation \eqref{eq: dual repres rho-tilde} becomes
\begin{equation} \label{eq: dual repres rho-tilde -R tilde}
\trho(X)= \sup_{Q \in \mathcal{Q}} \tilde{\mathfrak{R}}(\mathbb{E}_{Q} [\log X];Q),
\end{equation}
with $\tilde{\mathfrak{R}}:\mathbb{R}\times\mathcal{Q}\rightarrow[0,\infty]$ increasing and continuous from below in $t$,
while \eqref{eq: definition-R} becomes
\begin{equation} \label{eq: definition-R tilde}
\tilde{\mathfrak{R}}(t;Q) \triangleq \inf_{Y \in L^{\infty}_{++}} \{\trho(Y): \mathbb{E}_Q [\log Y] \geq t \}, \quad t \in \Bbb R, Q \in \mathcal{Q}.
\end{equation}

\bigskip

\begin{theorem} \label{thm: dual repres rho-tilde -star-shaped}
\textbf{[Monotonicity, quasi-logconvexity, and star-shapedness.
$\Leftrightarrow$
Representation with $\mathfrak{R}$ expansive in first coordinate.]}
\smallskip

The following statements are equivalent:
\begin{itemize}
\item[(a)] $\trho$ is monotone, quasi-logconvex, star-shaped and continuous from below;

\item[(b)] $\trho$ can be represented as in \eqref{eq: dual repres rho-tilde} with $\mathfrak{R}$ that is expansive in the first coordinate,
i.e., $\mathfrak{R}(t+h;Q) \geq \mathfrak{R}(t;Q) +h$ for any $t \in \Bbb R$, $h \in \Bbb R_+$ and $Q \in \mathcal{Q}$;

\item[(c)] $\trho$ can be represented as in \eqref{eq: dual repres rho-tilde -R tilde} with $\tilde{\mathfrak{R}}$
that is multiplicatively expansive in the first coordinate,
i.e., $\tilde{\mathfrak{R}}(t+h;Q) \geq e^h \tilde{\mathfrak{R}}(t;Q) $ for any $t \in \Bbb R$, $h \in \Bbb R_+$ and $Q \in \mathcal{Q}$.
\end{itemize}
\end{theorem}

\begin{proof}
\textit{(a)} $\Leftrightarrow$ \textit{(b)} 
By Theorem \ref{thm: dual repres rho-tilde}, 
it remains to prove that star-shapedness of $\trho$ is equivalent to expansivity of $R$.
Assume first that $\trho$ is star-shaped.
By \eqref{eq: definition-R}, it follows that
\begin{eqnarray*}
\mathfrak{R}(t+h;Q) &=& \inf \{\log (\trho(Y)): \mathbb{E}_Q [\log Y] \geq t+h \} \\
&=& \inf \{\log (\trho(Y)): \mathbb{E}_Q [\log Y -h] \geq t \} \\
&=& \inf \{\log (\trho(Y) ): \mathbb{E}_Q [\log (Y e^{-h})] \geq t \} \\
&=& \inf \{\log (\trho(\xi e^h) ): \mathbb{E}_Q [\log \xi] \geq t \} \\
&\geq & \inf \{\log (\trho(\xi ) )+h: \mathbb{E}_Q [\log \xi] \geq t \} \\
&=& \mathfrak{R}(t;Q) +h,
\end{eqnarray*}
for any $t \in \Bbb R$, $Q \in \mathcal{Q}$ and $h \geq 0$, where the inequality above is due to star-shapedness of $\rho$.

Assume now that $R$ is expansive.
By \eqref{eq: dual repres rho-tilde}, we deduce that
\begin{eqnarray*}
\trho( \lambda X) &=& \sup_{Q \in \mathcal{Q}} \{ \exp \left(\mathfrak{R}(\mathbb{E}_{Q} [\log (\lambda X)];Q) \right) \} \\
&=& \sup_{Q \in \mathcal{Q}} \{ \exp \left(\mathfrak{R}(\mathbb{E}_{Q} [\log X]+ \log \lambda;Q) \right) \} \\
& \geq & \sup_{Q \in \mathcal{Q}} \{ \exp \left(\mathfrak{R}(\mathbb{E}_{Q} [\log X];Q)+ \log \lambda \right) \} \\
& = & \lambda \sup_{Q \in \mathcal{Q}} \{ \exp \left(\mathfrak{R}(\mathbb{E}_{Q} [\log X];Q) \right) \} \\
&=& \lambda \trho (X),
\end{eqnarray*}
for any $\lambda \geq 1$ and $X \in L^{\infty}_{++}$, where the inequality above is due to expansivity of $\mathfrak{R}$.

\textit{(b)} $\Leftrightarrow$ \textit{(c)} Follows immediately by $\tilde{\mathfrak{R}}(t;Q)= \exp \left(\mathfrak{R}(t;Q)\right)$ and (a) $\Leftrightarrow$ (b).
\end{proof}

\bigskip

\begin{theorem} \label{thm: dual repres rho-tilde -PH}
\textbf{[Monotonicity, quasi-logconvexity, and positive homogeneity.
$\Leftrightarrow$
Representation with $\mathfrak{R}$ translation invariant in first coordinate $\Leftrightarrow$
Logconvex and positively homogeneous measures of risk.]}

The following statements are equivalent:

\begin{itemize}
\item[(a)] $\trho$ is monotone, quasi-logconvex, continuous from below and positively homogeneous;

\item[(a')] $\trho$ is logconvex and continuous from below;

\item[(b)] $\trho$ can be represented as in \eqref{eq: dual repres rho-tilde} with $\mathfrak{R}$ translation invariant in the first coordinate,
i.e., $\mathfrak{R}(t+h;Q) = \mathfrak{R}(t;Q) +h$ for any $t,h \in \Bbb R$ and $Q \in \mathcal{Q}$;

\item[(c)] $\trho$ can be represented as in \eqref{eq: dual repres rho-tilde -R tilde} with $\tilde{\mathfrak{R}}$
that is multiplicatively homogeneous in the first coordinate,
i.e., $\tilde{\mathfrak{R}}(t+h;Q) = e^h \tilde{\mathfrak{R}}(t;Q) $ for any $t,h \in \Bbb R$ and $Q \in \mathcal{Q}$.
\end{itemize}
\end{theorem}

\begin{proof}
\textit{(a) $\Leftrightarrow$ (a').} From Lemmas~\ref{lem:rho-trho} and \ref{lem: rho_trho-2},
it follows that (a) is equivalent to
$\rho(X)=\log(\trho(e^X))$ being a monotone, quasi-convex and continuous from below risk measure that is moreover translation invariant.
Hence, by 
Lemma \ref{lem:1a},
(a) is equivalent to $\rho$ being monotone, convex, translation invariant and continuous from below.
The thesis then follows from Lemmas \ref{lem:rho-trho} and \ref{lem: rho_trho-2}.

\textit{(a) $\Leftrightarrow$ (b) $\Leftrightarrow$ (c).} The proof can be established similarly
to that of Theorem \ref{thm: dual repres rho-tilde -star-shaped},
or by applying Theorem \ref{thm: dual repres rho-tilde -star-shaped} and Theorem 3.1 of Cerreia-Vioglio et al. \cite{CMMM11}.
Positive homogeneity is indeed equivalent to star-shapedness and $\trho( \lambda X)\leq \lambda \trho(X), \, \forall 0 \leq \lambda \leq 1, \forall X \in L^{\infty}_{++}$, together. Similarly as in Theorem \ref{thm: dual repres rho-tilde -star-shaped}, it can be checked that this last property corresponds to $\mathfrak{R}(t+h;Q) \leq \mathfrak{R}(t;Q) +h$ for any $t \in \Bbb R$, $h \in \Bbb R_+$ and $Q \in \mathcal{Q}$ or, equivalently, to $\tilde{\mathfrak{R}}(t+h;Q) \leq  e^h \tilde{\mathfrak{R}}(t;Q) $ for any $t \in \Bbb R$, $h \in \Bbb R_+$ and $Q \in \mathcal{Q}$.
\end{proof}

Note that if $\mathfrak{R}$ is cash-additive (or, equivalently, $\tilde{\mathfrak{R}}$ multiplicatively homogeneous
in the first coordinate), then
\begin{equation*}
\mathfrak{R}(t,Q)=\mathfrak{R}(0,Q) + t= t-c(Q), \quad t \in \Bbb R, Q \in \mathcal{Q},
\end{equation*}
where $c(Q)$ is the minimal penalty term of convex risk measures (see Frittelli and Maggis \cite{FM11}, Corollary~2.14).

\subsubsection{The building block}
Next, we reveal that a canonical Orlicz premium, given by the logarithmic certainty equivalent,
serves as the building block of logconvex and quasi-logconvex measures of risk.

We consider the following building block:
\begin{align*}
\exp \left ( \E \left [\log X \right ] \right ) &=
\inf \left \{ k >0 \; \Big | \; \E \left [  \log X - \log k  \right ] \leq 0 \right \}  \\
&= \inf \left \{ k >0 \; \Big | \; \E \left [ 1+ \log \left ( \frac{X}{k} \right)  \right ] \leq 1 \right \}  \\
&= \inf \left \{ k >0 \; \Big | \; \E \left [ \Phi \left ( \frac{X}{k} \right ) \right ] \leq 1 \right \} \\
&= H_\Phi (X),
\end{align*}
with $\Phi(x)=1+\log x$.
Henceforth, when $\Phi(x)=1+\log x$, we denote $H_\Phi (X)$ by $H_0 (X)$.

Then, for a logconvex measure of risk $\trho$ and its (monetary) convex counterpart $\rho$,
we have
\begin{align}
\trho(X) &= \exp\left( \rho(\log (X) \right) \nonumber \\
&= \exp\left( \sup_{Q \in \mathcal{Q}} \left\{\mathbb{E}_{Q} [\log(X) ]- c(Q)\right\} \right) \nonumber \\
&= \sup _{Q \in \mathcal{Q}} \beta(Q) \exp \left( \mathbb{E}_{Q} [\log(X) ] \right) \nonumber \\
&= \sup _{Q \in \mathcal{Q}} \beta(Q) H_{0,Q}(X),
\end{align}
where
\begin{equation*}
H_{\Phi, Q} (X)\triangleq \inf \left \{ k >0 \; \Big | \; \E_Q \left [ \Phi \left ( \frac{X}{k} \right ) \right ]  \leq 1 \right \},
\end{equation*}
and
\begin{equation}
H_{0, Q} (X)\triangleq \inf \left \{ k >0 \; \Big | \; \E_Q \left [ \log \left ( \frac{X}{k} \right ) \right ]  \leq 0 \right \},
\label{eq:buildingblockQ}
\end{equation}
and where $\beta(Q)=\exp(-c(Q))$.

Furthermore, for a quasi-logconvex measure of risk $\trho$,
we have
\begin{align}
\trho(X)&=\sup_{Q\in\mathcal{Q}}\exp \left(\mathfrak{R}\left(\mathbb{E}_{Q}\left[\log X\right]; Q\right)\right) \nonumber \\
&=\sup_{Q\in\mathcal{Q}}R\left(\exp\left(\mathbb{E}_{Q}\left[\log X\right]\right); Q\right) \nonumber \\
&=\sup_{Q\in\mathcal{Q}}R\left(H_{0, Q} (X); Q\right),
\label{eq:repQLC-1}
\end{align}
where
\begin{equation*}
\exp\left(\mathfrak{R}(t;Q)\right)
= R(\exp(t);Q),
\end{equation*}
hence
\begin{equation*}
R(s;Q)
=\exp\left(\mathfrak{R}(\log(s);Q)\right).
\end{equation*}

It follows that quasi-logconvex star-shaped measures of risk take the form \eqref{eq:repQLC-1} with $R$
geometrically expansive in the first coordinate,
that is,
\begin{equation}
R(s';Q)\geq R(s;Q)\exp\left(\vert \log s'-\log s\vert\right).
\label{eq:geoexpansive-1}
\end{equation}
Indeed,
\begin{equation*}
R(s';Q)=\exp(\mathfrak{R}(\log s';Q))\geq\exp(\mathfrak{R}(\log s;Q))\exp(|\log s'-\log s|).
\end{equation*}

\subsection{A Classification}

We state the following lemma, the proof of which follows from Lemmas \ref{lem:1a}, \ref{lem:rho-trho} and \ref{lem: rho_trho-2}
and the one-to-one correspondences \eqref{eq:trho-3}--\eqref{eq:trho-4}:
\begin{lemma}
\begin{itemize}
\item[(i)] Suppose $\trho$ is monotone, positively homogeneous and constant-multiplicative with $\trho(1)=1$.
Then $\trho$ is logconvex if and only if it is submultiplicative.
\item[(ii)] Suppose $\trho$ is monotone and positively homogeneous with $\trho(1)=1$, i.e., $\trho$ is a return risk measure.
Then $\trho$ is logconvex if and only if it is quasi-logconvex.
\end{itemize}
\label{lem:1b}
\end{lemma}

From (ii) in Lemma \ref{lem:1b} we conclude that, if positive homogeneity for return risk measures is replaced by star-shapedness,
then logconvexity should be replaced by quasi-logconvexity.

In Figure \ref{fig:II} we illustrate the connections between return, quasi-logconvex and related classes of risk measures.
The classes of logcoherent and logconvex risk measures are outside the class of monetary risk measures,
and (strict) subclasses of the class of return risk measures, as illustrated in panels (c) and (d); cf. (i) in Lemma \ref{lem:1b}.
From Theorems 6.1 and 6.2 in Laeven and Stadje \cite{LS13} and the one-to-one correspondences \eqref{eq:trho-3}--\eqref{eq:trho-4}
we obtain the decomposition of logconvex risk measures
into robust discounted $p$-norms, robust $p$-norms, and logcoherent measures of risk
illustrated in panels (e) and (f).
From (ii) in Lemma \ref{lem:1b} we conclude that the only intersection between quasi-logconvex and return risk measures
is given by the class of logconvex risk measures; see panel (g).
In panel (h) we illustrate the (strict) subclass of quasi-logconvex risk measures that are star-shaped.

\begin{figure}
\centering
\subfigure[Return]
{

\begin{tikzpicture}[scale = 0.34]
\draw[blue,thick,fill=blue!10!white,rotate around={45:(6,0)}] (6,0) rectangle ++(6,6);
\end{tikzpicture}

}
\ \ \ \ \ \ \ \ \ \
\subfigure[Monetary, Return and Coherent]
{

\begin{tikzpicture}[scale = 0.34]
\draw[blue,thick,fill=blue!10!white,rotate around={45:(6,0)}] (6,0) rectangle ++(6,6);
\draw[red,dashed,thick,rotate around={45:(1,0)}] (1,0) rectangle ++(6,6);
\draw[violet,dashed,thick,rotate around={45:(2.53,3.53)}] (2.53,3.53) rectangle ++(1,1);
\end{tikzpicture}

}
\subfigure[Monetary, Return, Coherent and Logcoherent]
{

\begin{tikzpicture}[scale = 0.34]
\draw[blue,thick,fill=blue!10!white,rotate around={45:(6,0)}] (6,0) rectangle ++(6,6);
\draw[red,dashed,thick,rotate around={45:(1,0)}] (1,0) rectangle ++(6,6);
\draw[violet,dashed,thick,rotate around={45:(2.53,3.53)}] (2.53,3.53) rectangle ++(1,1);
\draw[violet,thick,rotate around={45:(6.43,3.53)}] (6.43,3.53) rectangle ++(1,1);
\end{tikzpicture}

}
\ \ \ \ \
\subfigure[Monetary, Return, Coherent, Logcoherent and Logconvex]
{

\begin{tikzpicture}[scale = 0.34]
\draw[blue,thick,fill=blue!10!white,rotate around={45:(6,0)}] (6,0) rectangle ++(6,6);
\draw[red,dashed,thick,rotate around={45:(1,0)}] (1,0) rectangle ++(6,6);
\draw[violet,dashed,thick,rotate around={45:(2.53,3.53)}] (2.53,3.53) rectangle ++(1,1);
\draw[violet,thick,rotate around={45:(6.43,3.53)}] (6.43,3.53) rectangle ++(1,1);
\draw[green,thick,rotate around={45:(8,1.9)}] (8,1.9) rectangle ++(3.3,3.3);
\end{tikzpicture}

}
\subfigure[Monetary, Return, Coherent, Logcoherent, Logconvex and Robust Discounted $p$-Norm]
{

\begin{tikzpicture}[scale = 0.34]
\draw[blue,thick,fill=blue!10!white,rotate around={45:(6,0)}] (6,0) rectangle ++(6,6);
\draw[red,dashed,thick,rotate around={45:(1,0)}] (1,0) rectangle ++(6,6);
\draw[violet,dashed,thick,rotate around={45:(2.53,3.53)}] (2.53,3.53) rectangle ++(1,1);
\draw[violet,thick,rotate around={45:(6.43,3.53)}] (6.43,3.53) rectangle ++(1,1);
\draw[green,thick,rotate around={45:(8,1.9)}] (8,1.9) rectangle ++(3.3,3.3);
\draw[purple,thick,rotate around={45:(8.6,2.75)}] (8.6,2.75) rectangle ++(2.1,3);
\end{tikzpicture}

}
\ \ \ \ \
\subfigure[Monetary, Return, Coherent, Logcoherent, Logconvex, Robust Discounted $p$-Norm and Robust $p$-Norm]
{

\begin{tikzpicture}[scale = 0.34]
\draw[blue,thick,fill=blue!10!white,rotate around={45:(6,0)}] (6,0) rectangle ++(6,6);
\draw[red,dashed,thick,rotate around={45:(1,0)}] (1,0) rectangle ++(6,6);
\draw[violet,dashed,thick,rotate around={45:(2.53,3.53)}] (2.53,3.53) rectangle ++(1,1);
\draw[violet,thick,rotate around={45:(6.43,3.53)}] (6.43,3.53) rectangle ++(1,1);
\draw[green,thick,rotate around={45:(8,1.9)}] (8,1.9) rectangle ++(3.3,3.3);
\draw[purple,thick,rotate around={45:(8.6,2.75)}] (8.6,2.75) rectangle ++(2.1,3);
\draw[yellow,thick,rotate around={45:(7.1,4.3)}] (7.1,4.3) rectangle ++(2.1,0.8);
\end{tikzpicture}

}
\subfigure[
Return, Coherent, Logcoherent, Logconvex, Robust Discounted $p$-Norm, Robust $p$-Norm and Quasi-Logconvex]
{

\begin{tikzpicture}[scale = 0.34]
\draw[blue,thick,fill=blue!10!white,rotate around={45:(6,0)}] (6,0) rectangle ++(6,6);
\draw[violet,dashed,thick,rotate around={45:(2.53,3.53)}] (2.53,3.53) rectangle ++(1,1);
\draw[violet,thick,rotate around={45:(6.43,3.53)}] (6.43,3.53) rectangle ++(1,1);
\draw[green,thick,rotate around={45:(8,1.9)}] (8,1.9) rectangle ++(3.3,3.3);
\draw[purple,thick,rotate around={45:(8.6,2.75)}] (8.6,2.75) rectangle ++(2.1,3);
\draw[yellow,thick,rotate around={45:(7.1,4.3)}] (7.1,4.3) rectangle ++(2.1,0.8);
\draw[pink,thick,rotate around={45:(10,0)}] (10,0) rectangle ++(6,6);
\end{tikzpicture}

}
\ \ \
\subfigure[
Return, Coherent, Logcoherent, Logconvex, Robust Discounted $p$-Norm, Robust $p$-Norm, Quasi-Logconvex and Star-Shaped]
{

\begin{tikzpicture}[scale = 0.34]
\draw[blue,thick,fill=blue!10!white,rotate around={45:(6,0)}] (6,0) rectangle ++(6,6);
\draw[violet,dashed,thick,rotate around={45:(2.53,3.53)}] (2.53,3.53) rectangle ++(1,1);
\draw[violet,thick,rotate around={45:(6.43,3.53)}] (6.43,3.53) rectangle ++(1,1);
\draw[green,thick,rotate around={45:(8,1.9)}] (8,1.9) rectangle ++(3.3,3.3);
\draw[purple,thick,rotate around={45:(8.6,2.75)}] (8.6,2.75) rectangle ++(2.1,3);
\draw[yellow,thick,rotate around={45:(7.1,4.3)}] (7.1,4.3) rectangle ++(2.1,0.8);
\draw[pink,thick,rotate around={45:(10,0)}] (10,0) rectangle ++(6,6);
\draw[magenta,thick,rotate around={45:(9.1,0.8)}] (9.1,0.8) rectangle ++(4.9,4.9);
\end{tikzpicture}

}
\caption{Return, Quasi-Logconvex and Related Measures of Risk}
\small This figure illustrates the connections between return, quasi-logconvex and related classes of risk measures.
(We restrict to $L^{\infty}_{++}$.)
\label{fig:II}
\end{figure}

\subsubsection{On the intersection between quasi-convex risk measures and return risk measures}

In this subsection we show that quasi-convex risk measures can be positively homogeneous and yet not coherent.
This implies that the intersection between quasi-convex risk measures and return risk measures is not just given by the class of coherent risk measures and also gives rise to the question about the precise intersection between quasi-convex risk measures and quasi-logconvex risk measures. We investigate the relation between quasi-convex and return (or quasi-logconvex) risk measures below.

From Prop.~4.1 of Cerreia-Vioglio et al. \cite{CMMM11} we know that for a quasi-convex (and upper semicontinuous) risk measure $\rho$, positive homogeneity of $\rho$ is equivalent to positive homogeneity in $t$ of the functional $\mathfrak{R}(t,Q)$ in the dual representation of $\rho$.

Thanks to the above characterization, we are then able to prove the following result.

\begin{proposition}
If $\rho$ is a quasi-convex and monotone risk measure satisfying continuity from below and PH, then
\begin{equation} \label{eq: intersection-ph}
\rho(X)= \sup_{Q \in \mathcal{Q}} \{D_Q^+ \left(\mathbb{E}_{Q} [X]\right)^+ - D_Q^- \left(\mathbb{E}_{Q} [X]\right)^-  \}, \quad \mbox{ for any } X \in L^{\infty},
\end{equation}
for some $D_Q^+, D_Q^- \geq 0$ or, equivalently, the functional $\mathfrak{R}$ in the dual representation of $\rho$ has the following form:
\begin{equation} \label{eq: R-intersection-ph}
\mathfrak{R}(t,Q)= D_Q^+ \, t^+ - \, D_Q^- t^-, \quad \mbox{ for any } t \in \mathbb{R}.
\end{equation}
\end{proposition}

\begin{proof}
By Prop.~2.13 of Frittelli and Maggis \cite{FM11},  Cerreia-Vioglio et al. \cite{CMMM11}, Theorem 3.1 and Lemma 3.2,
and Drapeau and Kupper \cite{DK13}, $\rho$ has the following dual representation:
\begin{equation} \label{eq: dual repres rho-1}
\rho(X)=\sup_{Q\in\mathcal{Q}} \mathfrak{R}\left(\mathbb{E}_{Q}\left[ X\right]; Q\right),
\end{equation}
where $\mathfrak{R}:\mathbb{R}\times\mathcal{Q}\rightarrow[-\infty,\infty]$ is
monotone increasing and continuous from below in $t \in \mathbb{R}$. Moreover, by (proceeding similarly as in) Prop. 4.1 of Cerreia-Vioglio et al. \cite{CMMM11}, PH of $\rho$ is equivalent to positive homogeneity of $\mathfrak{R}(t,Q)$ in $t$ for any $Q$. It follows therefore that for any fixed $Q \in \mathcal{Q}$ there exist some $D_Q^+, D_Q^- \geq 0$ such that
\begin{equation*}
\mathfrak{R}(t,Q)= D_Q^+ \, t^+ - D_Q^- \, t^-, \quad \mbox{ for any } t \in \mathbb{R}.
\end{equation*}
\end{proof}

The previous result implies that the intersection between quasi-convex risk measures and return risk measures (with $\rho(1)=1$) satisfying continuity from below is not just given by the class of coherent risk measures, but by risk measures as in \eqref{eq: intersection-ph} (with $\sup_{Q \in \mathcal{Q}} D_Q^+=1$). The class of coherent risk measures corresponds to the particular case where $D_Q^+=D_Q^-=1$ for any $Q \in \mathcal{Q}$.

Note that if $D_Q^+=D_Q^-=D_Q$ for any $Q \in \mathcal{Q}$ but not identically equal to $1$, then
\begin{equation*}
\rho(X)= \sup_{Q \in \mathcal{Q}} \{D_Q \mathbb{E}_{Q} [X]\}=\sup_{Q \in \mathcal{Q}} \mathbb{E}_{\mu_Q} [X],
\end{equation*}
where $\mu_Q= D_Q \cdot Q$ is a measure with $\mu_Q (\Omega)=D_Q$, that is not necessarily normalized. The case where $D_Q \leq 1$ for any $Q \in \mathcal{Q}$ corresponds to sublinear risk measures (see Frittelli \cite{F00} and Frittelli and Rosazza Gianin \cite{FR02}).
\bigskip

Below we provide an example of a quasi-logconvex measure that is not quasi-convex.

\begin{example}[QLC but not QC] \label{ex: qlc not qc}
Consider $\rho \colon L^{\infty}_{++} \to [0, \infty]$ defined as
\begin{equation} \label{eq: example - qlc not qc}
\rho(X)=\sup_{Q\in\mathcal{Q}} \exp\left(\mathbb{E}_{Q}\left[ \log X\right]\right)
\end{equation}
for a given set of probability measures $\mathcal{Q}$.
In other words, $\rho$ is a logcoherent risk measure.

It follows easily that $\rho$ is monotone and $\rho(1)=1$.
Furthermore, it is also quasi-logconvex.
Indeed, for any $\alpha \in [0,1]$ and $X,Y \in L^{\infty}_{++}$,
\begin{eqnarray*}
\rho(X^{\alpha} Y^{1-\alpha})&=&\sup_{Q\in\mathcal{Q}} \exp\left(\mathbb{E}_{Q}\left[ \log (X^{\alpha} Y^{1-\alpha})\right]\right) \\
&=&\sup_{Q\in\mathcal{Q}} \exp\left(\alpha\mathbb{E}_{Q}\left[ \log X \right]+ (1-\alpha)\mathbb{E}_{Q}\left[ \log Y \right]\right) \\
&\leq &\sup_{Q\in\mathcal{Q}} \exp\left(\max\{\mathbb{E}_{Q}\left[ \log X\right];\mathbb{E}_{Q}\left[ \log Y\right]\}  \right) \\
&\leq &\sup_{Q\in\mathcal{Q}} \max\{\exp\left(\mathbb{E}_{Q}\left[ \log X\right]\right);\exp\left(\mathbb{E}_{Q}\left[ \log Y\right]\right)\} \\
&= &\max \{\rho(X); \rho(Y) \}.
\end{eqnarray*}

Although $\rho$ is quasi-logconvex, it is not quasi-convex in general. Take, for instance, $\Omega=\{\omega_1, \omega_2\}$, $\mathcal{Q}=\{ \bar{Q} \}$ with $\bar{Q}(\omega_1)=\bar{Q}(\omega_2)= \frac 12$,
$X= \left\{
\begin{array}{rl}
1;& \omega_1 \\
e^3;& \omega_2
\end{array}
\right.$,
$Y=e^2$ and $\alpha= \frac 12$.
It then follows that $\rho(X)= \exp\left(\mathbb{E}_{\bar{Q}}\left[ \log X\right]\right)$ and that
\begin{eqnarray*}
\rho\left(\frac{X+ Y}{2}\right)&=& \exp\left(\mathbb{E}_{\bar{Q}}\left[ \log \left(\frac{X+ Y}{2}\right)\right]\right) \\
&=& \exp\left( \frac 12 \log \left(\frac{1+ e^2}{2}\right) + \frac 12 \log \left(\frac{e^3+ e^2}{2}\right) \right) \\
&>& \max \{\rho(X); \rho(Y) \} =e^2.
\end{eqnarray*}
Hence, $\rho$ fails to satisfy quasi-convexity.
%
\end{example}

\begin{example}[QLC but not QC and not PH]
Consider $\rho \colon L^{\infty}_{++} \to [0, \infty]$ defined as
\begin{equation} \label{eq: example - qlc not qc not ph}
\rho(X)=\sup_{Q\in\mathcal{Q}} \exp\left(\ell^{-1} \left(\mathbb{E}_{Q}\left[\ell( \log X)\right]\right)\right),
\end{equation}
for a given set of probability measures $\mathcal{Q}$ and for a strictly increasing and convex function $\ell: \mathbb{R} \to \mathbb{R}$ with $\ell(0)=0$. Note that Example \ref{ex: qlc not qc} is a particular case of the present one and that $\ell^{-1} \left(\mathbb{E}_{Q}\left[\ell( X)\right]\right)$ corresponds to the mean value premium principle that, under the assumptions on $\ell$, is known to be quasi-convex and monotone (see, e.g., \cite{CMMM11} and \cite{FM11b}).

It follows easily that $\rho$ is monotone and $\rho(1)=1$. Furthermore, it is also quasi-logconvex. Indeed, for any $\alpha \in [0,1]$ and $X,Y \in L^{\infty}_{++}$,
\begin{eqnarray*}
\rho(X^{\alpha} Y^{1-\alpha})&=&\sup_{Q\in\mathcal{Q}} \exp\left(\ell^{-1} \left(\mathbb{E}_{Q}\left[\ell( \log (X^{\alpha} Y^{1-\alpha}))\right]\right)\right) \\
&=&\sup_{Q\in\mathcal{Q}} \exp\left(\ell^{-1} \left(\mathbb{E}_{Q}\left[\ell( \alpha \log X + (1-\alpha) \log Y)\right]\right)\right) \\
&\leq &\sup_{Q\in\mathcal{Q}} \exp\left(\max\{\ell^{-1} \left(\mathbb{E}_{Q}\left[\ell( \log X)\right]\right);\ell^{-1} \left(\mathbb{E}_{Q}\left[\ell( \log Y)\right]\right)\}  \right) \\
&= &\sup_{Q\in\mathcal{Q}} \max\{\exp\left(\ell^{-1} \left(\mathbb{E}_{Q}\left[\ell( \log X)\right]\right)\right);\exp\left(\ell^{-1} \left(\mathbb{E}_{Q}\left[\ell( \log Y)\right]\right)\right)\} \\
&= &\max \{\rho(X); \rho(Y) \},
\end{eqnarray*}
where the inequality above is due to quasi-convexity of the mean value premium principle $\ell^{-1} \left(\mathbb{E}_{Q}\left[\ell( X)\right]\right)$.

Although $\rho$ is quasi-logconvex, it is not quasi-convex in general. Take, for instance, $\Omega=\{\omega_1, \omega_2\}$, $\mathcal{Q}=\{ \bar{Q} \}$ with $\bar{Q}(\omega_1)=\bar{Q}(\omega_2)= \frac 12$,
$X= \left\{
\begin{array}{rl}
1;& \omega_1 \\
e^3;& \omega_2
\end{array}
\right.$,
$Y=e^2$, $\alpha= \frac 12$, and $\ell(x)= \left\{
\begin{array}{rl}
x;& x<0 \\
x^2 +x;& x \geq 0
\end{array}
\right.$.
It then follows that $\rho(X)= \exp\left(\ell^{-1} \left(\mathbb{E}_{\bar{Q}}\left[\ell( \log X)\right]\right)\right)$ and that
\begin{eqnarray*}
\rho\left(\frac{X+ Y}{2}\right)&=& \exp\left(\ell^{-1} \left(\mathbb{E}_{\bar{Q}}\left[ \ell\left(\log \left(\frac{X+ Y}{2}\right)\right)\right] \right)\right) \\
&=& \exp\left( \ell^{-1} \left(\frac 12 \, \ell \left(\log \left(\frac{1+ e^2}{2}\right)\right) + \frac 12 \, \ell \left(\log \left(\frac{e^3+ e^2}{2}\right) \right)\right) \right) \\
&>& \max \{\rho(X); \rho(Y) \} =e^2.
\end{eqnarray*}
Hence, $\rho$ fails to satisfy quasi-convexity.

Furthermore, $\rho$ also fails to be positively homogeneous. 
By taking the same $\Omega$, $\mathcal{Q}$, $X$ and $\ell$ as before and $\lambda=\frac{1}{e}$, indeed,
\begin{eqnarray*}
\rho(\lambda X)&=& \exp\left(\ell^{-1} \left(\mathbb{E}_{\bar{Q}}\left[\ell( \log (\lambda X)\right]\right)\right) \\
&=& \exp\left(\ell^{-1} \left(\frac12 \, \ell(-1)+\frac12 \, \ell(2)\right)\right)\\
&=& \exp\left( \ell^{-1} \left(\frac52 \right)\right) \\
&=& \exp\left( \frac{\sqrt{11}-1}{2}\right) > \lambda \rho(X)=e.
\end{eqnarray*}

Note that if the mean value premium principle is cash-superadditive, then $\rho$ in \eqref{eq: example - qlc not qc not ph} satisfies star-shapedness. In this case, indeed, for any $\lambda \geq 1$ and $X \in L^{\infty}_{++}$ it would hold that
\begin{eqnarray*}
\rho(\lambda X)&=& \sup_{Q \in\mathcal{Q}} \exp\left(\ell^{-1} \left(\mathbb{E}_{Q}\left[\ell( \log  X+ \log \lambda \right]\right)\right) \\
&\geq & \sup_{Q \in\mathcal{Q}} \exp\left(\ell^{-1} \left(\mathbb{E}_{Q}\left[\ell( \log  X \right] + \log \lambda\right)\right) \\
&=& \lambda \rho(X).
\end{eqnarray*}
Positive homogeneity of the risk measure of Example \ref{ex: qlc not qc} then follows by cash-additivity of the corresponding mean value premium principle.
\end{example}

Viceversa, any quasi-convex risk measure (restricted to $L^{\infty}_{++}$) is also quasi-logconvex. More precisely, the following result replies to the question about the intersection between quasi-convex risk measures and quasi-logconvex risk measures.

\begin{proposition} \label{prop: intersect qco-qlc}
Let $\rho:L^{\infty} \to [-\infty; \infty]$ be a monotone and continuous from below risk measure with $\rho(0)=0$.

If $\rho$ is quasi-convex, then its restriction to $L^{\infty}_{++}$ is also quasi-logconvex.
\end{proposition}

\begin{proof}
By Prop.~2.13 of Frittelli and Maggis \cite{FM11}, Cerreia-Vioglio et al. \cite{CMMM11}, Theorem~3.1 and Lemma~3.2,
and Drapeau and Kupper \cite{DK13}, $\rho$ has the following dual representation:
\begin{equation*}
\rho(X)=\sup_{Q\in\mathcal{Q}} \mathfrak{R}\left(\mathbb{E}_{Q}\left[ X\right]; Q\right),
\end{equation*}
where $\mathfrak{R}:\mathbb{R}\times\mathcal{Q}\rightarrow[-\infty,\infty]$ is
monotone increasing and continuous from below in $t \in \Bbb R$.

For any $X,Y \in L^{\infty}_{++}$ and $\alpha \in (0,1)$, it follows that
\begin{eqnarray}
\rho \left(X^{\alpha}Y^{1-\alpha}\right) &=& \sup_{Q\in\mathcal{Q}} \mathfrak{R}\left(\mathbb{E}_{Q}\left[ X^{\alpha}Y^{1-\alpha}\right]; Q\right) \notag \\
& \leq & \sup_{Q\in\mathcal{Q}} \mathfrak{R}\left( (\mathbb{E}_{Q}\left[ X\right])^{\alpha} (\mathbb{E}_{Q}\left[ Y\right])^{1-\alpha}; Q\right) \label{eq: holder}\\
& \leq & \sup_{Q\in\mathcal{Q}} \mathfrak{R}\left( \mathbb{E}_{Q}\left[ X\right] \vee \mathbb{E}_{Q}\left[ Y\right]; Q\right) \label{eq: ineq R000}\\
&=& \sup_{Q\in\mathcal{Q}} \left\{ \mathfrak{R}\left( \mathbb{E}_{Q}\left[ X\right] ; Q\right) \vee \mathfrak{R}\left( \mathbb{E}_{Q}\left[ Y\right]; Q\right) \right\} \label{eq: eq R000}\\
&=& \rho(X) \vee \rho(Y), \notag
\end{eqnarray}
where \eqref{eq: holder} is due to the H\"{o}lder inequality while \eqref{eq: ineq R000} and \eqref{eq: eq R000} follow by increasing monotonicity of $\mathfrak{R}$.
\end{proof}
As a consequence of the previous result, the families of coherent and of convex risk measures that are continuous from below and satisfy $\rho(0)=0$ are quasi-logconvex once restricted to $L^{\infty}_{++}$.

\setcounter{equation}{0}

\section{Acceptance Sets}\label{sec:acc}


Similarly to the case of convex monetary risk measures,
Drapeau and Kupper \cite{DK13}, Theorem~1, proved a one-to-one correspondence between quasi-convex risk measures and families of acceptance sets
referred to as risk acceptance families,
that is, they characterized quasi-convex risk measures in terms of properties of families of acceptance sets. 
Note that in this paper we use a convention on signs and on (increasing) monotonicity of $\rho$ that is different from \cite{DK13}. 
This different convention will be reflected in the definition of (monotonicity for) acceptance sets as well as in the representation of $\rho$ in terms of acceptance sets.\smallskip

We recall from \cite{DK13} that a \textit{risk acceptance family} is a family $(\mathcal{A}^a)_{a \in \Bbb R}$, with $\mathcal{A}^a \subseteq L^{\infty}$ for any $a \in \Bbb R$, that is increasing in $a$ and satisfies
\begin{itemize}
\item[(i)] $\mathcal{A}^a$ is convex for any $a \in \Bbb R$;
\item[(ii)] $\mathcal{A}^a$ is monotone for any $a \in \Bbb R$, i.e., if $Y \geq X$ and $X \in \mathcal{A}^a$, then also $Y \in \mathcal{A}^a$;
\item[(iii)] right-continuity: $\mathcal{A}^a = \bigcap_{\bar{a} >a} \mathcal{A}^{\bar{a}}$ for any $a \in \Bbb R$.
\end{itemize}

As shown in Drapeau and Kupper \cite{DK13}, Theorem~1 and Propositions~2 and~3, given a risk measure $\rho$ and the family $(\mathcal{A}^a)_{a \in \Bbb R}$ defined by $\mathcal{A}^a=\mathcal{A}_{\rho}^a \triangleq \{Y \in L^{\infty}: \rho(-Y) \leq a\}$, it holds that \medskip

\noindent quasi-convex and monotone $\rho$ $\longleftrightarrow$ risk acceptance family $(\mathcal{A}^a)_{a \in \Bbb R}$; \medskip

\noindent ... plus cash-subadditive $\longleftrightarrow$ ... plus $\mathcal{A}^a \subseteq \mathcal{A}^{a+k} + k$ for any $a \in \Bbb R$ and $k \geq 0$;
\medskip

\noindent ... plus cash-additive (hence convex) $\longleftrightarrow$ ... plus $\mathcal{A}^0 = \mathcal{A}^{h} + h$ for any $h \in \Bbb R$.
\medskip

Furthermore, $\rho$ can be represented in terms of the family $(\mathcal{A}^a)_{a \in \Bbb R}$ as
\begin{equation*}
\rho(X)=\inf\{a \in \Bbb R: -X \in \mathcal{A}^a \} \quad \mbox{ for any } X \in L^{\infty}.
\end{equation*}
\smallskip

By the one-to-one correspondence \eqref{eq:trho-3}--\eqref{eq:trho-4},
the result above can be extended to quasi-logconvex risk measures.

For any $a \in \Bbb R$ and $b \in \Bbb R_+$, denote by
\begin{equation} \label{eq: accept sets}
\mathcal{A}_{\rho} ^a \triangleq \{Y \in L^{\infty}: \rho(-Y) \leq a\}, \quad \mathcal{B}_{\trho} ^{b} \triangleq \left\{X \in L^{\infty}_{++}: \trho\left(\frac{1}{X} \right) \leq b \right\},
\end{equation}
the acceptance sets of $\rho$ at the level $a$ and of $\trho$ at the level $b$, respectively.

It is easy to verify that
\begin{equation} \label{eq: relation accept sets}
\mathcal{A}_{\rho}^a= \left\{ \log Z: \trho\left(\frac{1}{Z} \right) \leq e^a \right\},
\end{equation}
hence
\begin{equation} \label{eq: relation accept sets-2}
Y \in \mathcal{A}_{\rho}^a \Leftrightarrow  e^Y \in \mathcal{B}_{\trho} ^{e^a}.
\end{equation}

Note that the multiplicative nature of the family $(\mathcal{B}^b)_{b \in \Bbb R_+}$ is not surprising due to the interpretation of geometric risk measures.
\medskip

We then define a \textit{log-risk acceptance family} to be a family $(\mathcal{B}^b)_{b \in \Bbb R_+}$, with $\mathcal{B}^b \subseteq L^{\infty}_{++}$ for any $b \in \Bbb R_+$, that is increasing in $b$ and satisfies
\begin{itemize}
\item[(i)] $(\mathcal{B}^b)_{b \in \Bbb R_+}$ is log-convex for any $b \in \Bbb R_+$,
i.e., if $\frac{1}{X}, \frac{1}{Y} \in \mathcal{B}^b$, then also $\frac{1}{X^{\alpha} Y^{1 - \alpha}} \in \mathcal{B}^b$ for any $\alpha \in [0,1]$;
\item[(ii)] $\mathcal{B}^b$ is monotone for any $b \in \Bbb R_+$;
\item[(iii)] right-continuity: $\mathcal{B}^b = \bigcap_{\bar{b} >b} \mathcal{B}^{\bar{b}}$ for any $b \in \Bbb R_+$.
\end{itemize}

Furthermore, the log-risk acceptance family is said to be:
\begin{itemize}
\item[(iv)] \noindent \textit{B-star-shaped} if: $X \in \mathcal{B}^{b}$ $\Rightarrow$ $\lambda X \in \mathcal{B}^{b/ \lambda}$ for any $\lambda \geq 1$;
\item[(v)] \noindent \textit{B-positively homogeneous} if: $\gamma X \in \mathcal{B}^{b/ \gamma}$ (with $\gamma >0$) $\Leftrightarrow$ $X \in \mathcal{B}^{b}$.
\end{itemize}

Note that B-positive homogeneity implies B-star-shapedness.
\medskip

The following theorems provide a one-to-one correspondence between quasi-logconvex risk measures and log-risk acceptance families.

\begin{theorem} \label{thm: risk acceptance - trho} \textbf{[Monotone and quasi-logconvex $\trho$
$\Leftrightarrow$
Log-risk acceptance families.]}
\begin{itemize}
\item[(a)] Given a quasi-logconvex and monotone risk measure $\trho$,
$(\mathcal{B}_{\trho} ^{b})_{b \in \Bbb R_+}$ defined in \eqref{eq: accept sets} is a log-risk acceptance family.

\item[(b)] Viceversa: given a log-risk acceptance family $(\mathcal{B} ^{b})_{b \in \Bbb R_+}$,
\begin{equation} \label{eq: rho-tilde from acceptance}
\trho_{\mathcal{B}}(X)= \inf \left\{b \in \Bbb R_+: \frac{1}{X} \in \mathcal{B} ^{b} \right\}
\end{equation}
is a quasi-logconvex and monotone risk measure.
\end{itemize}
\end{theorem}

\begin{proof}

(a) Monotonicity in $b$ of $\mathcal{B}^b$ as well as monotonicity (ii) and right-continuity (iii) follow immediately by the properties of $\trho$ and \eqref{eq: accept sets}. Log-convexity (i) follows by quasi-logconvexity of $\trho$.

(b) Monotonicity of $\trho_{\mathcal{B}}$ follows immediately by monotonicity of any set $\mathcal{B}^b$.
Indeed, if $X \geq Y$, then
$$
\trho_{\mathcal{B}}(X)= \inf \left\{b \in \Bbb R_+: \frac{1}{X} \in \mathcal{B} ^{b} \right\} \geq  \inf \left\{b \in \Bbb R_+: \frac{1}{Y} \in \mathcal{B} ^{b} \right\}= \trho_{\mathcal{B}}(Y).
$$
It remains to prove quasi-logconvexity of $\trho$.
Given $X,Y \in L^{\infty}_{++}$, set $\bar{b}=\max \{\trho_{\mathcal{B}}(X); \trho_{\mathcal{B}}(Y)\}$. Hence, $\frac{1}{X},\frac{1}{Y} \in \mathcal{B} ^{\bar{b}}$.
By log-convexity of any $\mathcal{B}^b$ it follows that
$$
\trho_{\mathcal{B}}(X^{\alpha} Y^{1 - \alpha})= \inf \left\{b \in \Bbb R_+: \frac{1}{X^{\alpha} Y^{1 - \alpha}} \in \mathcal{B} ^{b} \right\} \leq \bar{b}=\max \{\trho_{\mathcal{B}}(X); \trho_{\mathcal{B}}(Y)\},
$$
that is, quasi-logconvexity of $\trho$.
\end{proof}

\begin{theorem} \label{thm: risk acceptance - trho  - star-shaped} \textbf{[Monotone, quasi-logconvex and star-shaped $\trho$
$\Leftrightarrow$ B-star-shaped
log-risk acceptance families.]}
\begin{itemize}
\item[(a)] Given a quasi-logconvex, monotone and star-shaped risk measure $\trho$,
$(\mathcal{B}_{\trho} ^{b})_{b \in \Bbb R_+}$ defined in \eqref{eq: accept sets} is a log-risk acceptance family satisfying B-star-shapedness.

\item[(b)] Viceversa: given a log-risk acceptance family $(\mathcal{B} ^{b})_{b \in \Bbb R_+}$ satisfying B-star-shapedness,
$\trho_{\mathcal{B}}$ defined in \eqref{eq: rho-tilde from acceptance}
is a quasi-logconvex, monotone and star-shaped risk measure.
\end{itemize}
\end{theorem}

\begin{proof} By Theorem \ref{thm: risk acceptance - trho},
it remains to prove that star-shapedness of $\trho$ corresponds to B-star-shapedness of $(\mathcal{B}_{\trho} ^{b})_{b \in \Bbb R_+}$. Recall that star-shapedness of $\trho$ can be equivalently formulated as $\trho(\beta X) \leq \beta\trho(X)$ for any $X \in L^{\infty}_{++}$ and $0<\beta\leq 1$.
\smallskip

(a) If $X \in \mathcal{B}^{b}$,
then $ \trho\left(\frac{1}{\lambda X}\right) \leq \frac{1}{\lambda} \trho(\frac{1}{X}) \leq \frac{b}{\lambda}$ for any $\lambda \geq 1$ (by star-shapedness of $\trho$).
Hence, $\lambda X \in \mathcal{B}^{b/ \lambda}$ for any $\lambda \geq 1$,
that is, B-star-shapedness of $(\mathcal{B}_{\trho} ^{b})_{b \in \Bbb R_+}$.

(b) By B-star-shapedness, it follows that for any $\lambda \geq 1$ and $X \in L^{\infty}_{++}$
\begin{eqnarray*}
\trho_{\mathcal{B}}(\lambda X) &=& \inf \left\{b \in \Bbb R_+: \frac{1}{\lambda X} \in \mathcal{B} ^{b} \right\} \\
& \geq & \inf \left\{b \in \Bbb R_+: \frac{1}{X} \in \mathcal{B} ^{b/\lambda} \right\} \\
& = & \lambda \inf \left\{ \bar{b} \in \Bbb R_+: \frac{1}{X} \in \mathcal{B} ^{\bar{b}} \right\} \\
&=& \lambda \trho_{\mathcal{B}}(X).
\end{eqnarray*}
\end{proof}

\begin{theorem} \label{thm: risk acceptance - trho  - PH} \textbf{[Monotone, quasi-logconvex and positively homogeneous $\trho$
$\Leftrightarrow$
B-positively homogeneous log-risk acceptance families.]}
\begin{itemize}
\item[(a)] Given a quasi-logconvex, monotone and positively homogeneous risk measure $\trho$,
$(\mathcal{B}_{\trho} ^{b})_{b \in \Bbb R_+}$ defined in \eqref{eq: accept sets} is a log-risk acceptance family satisfying B-positive homogeneity.

\item[(b)] Viceversa: given a log-risk acceptance family $(\mathcal{B} ^{b})_{b \in \Bbb R_+}$ satisfying B-positive homogeneity,
$\trho_{\mathcal{B}}$ defined in \eqref{eq: rho-tilde from acceptance}
is a quasi-logconvex, monotone and positively homogeneous risk measure.
\end{itemize}
\end{theorem}

\begin{proof}
The proof can be established similarly to that of Theorem \ref{thm: risk acceptance - trho  - star-shaped}.
We include the proof for completeness.

(a) If $\gamma X \in \mathcal{B}^{b/\gamma}$ for some $\gamma >0$,
then $\trho \left(\frac{1}{X}\right) = \gamma \cdot \trho \left(\frac{1}{\gamma X}\right) \leq b$ (by positive homogeneity of $\trho$).
Hence, $X \in \mathcal{B}^{b}$.

Viceversa: if $X \in \mathcal{B}^{b}$ then, again by positive homogeneity of $\trho$, $\trho \left(\frac{1}{\gamma X}\right)= \frac{1}{\gamma} \trho \left(\frac{1}{X} \right) \leq \frac{b}{\gamma}$ for any $\gamma >0$. Hence $\gamma X \in \mathcal{B}^{b/\gamma}$ for any $\gamma >0$.

B-positive homogeneity of $(\mathcal{B}_{\trho}^{b})_{b \in \Bbb R_+}$ is therefore established.

(b) By B-positive homogeneity, it follows that for any $\gamma >0$ and $X \in L^{\infty}_{++}$,
\begin{eqnarray*}
\trho_{\mathcal{B}}(\gamma X)&=&\inf \left\{b \in \Bbb R_+: \frac{1}{\gamma X} \in \mathcal{B} ^{b} \right\} \\
&=&\inf \left\{b \in \Bbb R_+: \frac{1}{X} \in \mathcal{B} ^{b/\gamma} \right\} \\
&=&\inf \left\{b \gamma \in \Bbb R_+: \frac{1}{X} \in \mathcal{B} ^{b} \right\} \\
&=&\gamma \trho_{\mathcal{B}}(X).
\end{eqnarray*}
\end{proof}

%

\setcounter{equation}{0}

\section{Law-Invariant Representation}\label{sec:li}



In this section, we will focus on quasi-logconvex risk measures that are also law invariant.
In particular, we provide their dual representation in terms of Orlicz premia as building blocks or in terms of AR@R,
by applying the dual representation of quasi-convex risk measures in Cerreia-Vioglio, Maccheroni, Marinacci and Montrucchio \cite{CMMM11}.

To this aim, we recall that a risk measure $\rho\colon L^{\infty} \to [-\infty,\infty]$ is said to be
\begin{itemize}
\item [-]law invariant if $X \overset{d}{\sim} Y \Rightarrow \rho(X) = \rho (Y)$,
\end{itemize}
where $X \overset{d}{\sim} Y$ means that $X$ and $Y$ have the same distribution with respect to $P$. For law invariant risk measures we refer, e.g., to Kusuoka \cite{K01}, Frittelli and Rosazza Gianin \cite{FR05}, Jouini et al. \cite{JST08}, and Cerreia-Vioglio et al. \cite{CMMM11}.


\begin{lemma} \label{lem: law invar trho-rho}
Let $\rho \colon L^{\infty} \to [-\infty,\infty]$ be a risk measure and
let the associated $\trho \colon L^{\infty}_{++} \to [0, \infty]$ be given by
\eqref{eq:trho-3}.

$\rho$ is law invariant if and only if $\trho$ is law invariant.
\end{lemma}

\begin{proof}
\textit{Only if case.} For any $X,Y \in L^{\infty}_{++}$ with $X \overset{d}{\sim} Y$, i.e., with the same distribution, it holds that also $\log X \overset{d}{\sim} \log Y$.
By law invariance of $\rho$, it follows that
\begin{equation*}
\trho(X)=\exp \left( \rho \left( \log(X) \right) \right)=\exp \left( \rho \left( \log(Y) \right) \right)=\trho(Y).
\end{equation*}

\textit{If case.} It can be proved similarly.
\end{proof}
\medskip

Set $\mathcal{R}\triangleq \tilde{\mathfrak{R}}=\exp(\mathfrak{R})$.

\begin{theorem} \label{thm: law invar-repres}
\textbf{[Dual representation of law invariant quasi-logconvex risk measures]}
\smallskip
Let $(\Omega, \mathcal{F},P)$ be a non-atomic probability space.

(a) $\trho \colon L^{\infty}_{++} \to [0, \infty]$ is monotone, quasi-logconvex, continuous from below and law invariant if and only if  $\mathcal{R} :\mathbb{R}\times\mathcal{Q}\rightarrow[0,\infty]$ in \eqref{eq: dual repres rho-tilde -R tilde}--\eqref{eq: definition-R tilde} is law invariant, i.e., $\mathcal{R}(t, \cdot)$ is law invariant on $\mathcal{Q}$ for any $t \in \Bbb R$.

(b) If $\trho \colon L^{\infty}_{++} \to [0, \infty]$ is monotone, quasi-logconvex, continuous from below and law invariant,
then it has the following dual representation:
\begin{equation} \label{eq: dual repres rho-tilde-law inv}
\trho(X)= \sup_{Q \in \mathcal{Q}} \mathcal{R} \left(\int_0^1 \log (q_X (\beta)) \, q_{\frac{\mathrm{d}Q}{\mathrm{d}P}} (\beta) \,\mathrm{d}\beta;Q \right),
\end{equation}
where $q_Z(\beta)$ is any $\beta$-quantile of $Z$ and
\begin{equation} \label{eq: definition-tR-law inv}
\mathcal{R}(t;Q) \triangleq \inf_{ Y \in L^{\infty}_{++}} \left\{ \trho(Y): \int_0^1 \log (q_Y (1-\beta)) \, q_{\frac{\mathrm{d}Q}{\mathrm{d}P}} (\beta) \,\mathrm{d}\beta \geq t \right\}, \quad t \in \Bbb R,\quad Q \in \mathcal{Q},
\end{equation}
is law invariant.
\end{theorem}

\begin{proof}
The result can be proved by applying the one-to-one correspondence between $\rho$ and $\trho$ and Theorem 5.1 of Cerreia-Vioglio, Maccheroni, Marinacci and Montrucchio \cite{CMMM11}. Notice that the proof of this result holds even if $\mathfrak{R}$ is defined as in \eqref{eq: definition-R} (see footnote 21 in \cite{CMMM11}) and if the supremum over $Q$ is not attained in the dual representation of quasi-convex risk measures (or under continuity from above instead of below).
\medskip

(a) follows immediately by Theorem 5.1 of Cerreia-Vioglio, Maccheroni, Marinacci and Montrucchio \cite{CMMM11}, by the one-to-one correspondence between $\rho$ and $\trho$ as well as $\mathfrak{R}$ and $\mathcal{R}=\exp(\mathfrak{R})$ and by the equivalence between law invariance of $\rho$ (resp.~of $\mathfrak{R}$) and that of $\trho$ (resp.~of $\mathcal{R}$ ).

(b)
Since $\trho$ is monotone, quasi-logconvex, continuous from below and law invariant, the associated $\rho$ is monotone, quasi-convex, continuous from below and law invariant (see Lemmas \ref{lem:rho-trho}, \ref{lem: rho_trho-2} and \ref{lem: law invar trho-rho}). By Theorem 5.1 of Cerreia-Vioglio, Maccheroni, Marinacci and Montrucchio \cite{CMMM11},
\begin{equation} \label{eq: R-law inv}
\rho(X)= \sup_{Q \in \mathcal{Q}} \mathfrak{R} \left(\int_0^1 q_X (\beta) \, q_{\frac{\mathrm{d}Q}{\mathrm{d}P}} (\beta) \,\mathrm{d}\beta;Q \right),
\end{equation}
where
\begin{equation} \label{eq: definition-R-law inv}
\mathfrak{R}(t;Q) \triangleq \inf_{ Z \in L^{\infty}} \left\{ \rho(Z): \int_0^1 q_Z (1-\beta) \, q_{\frac{\mathrm{d}Q}{\mathrm{d}P}} (\beta) \,\mathrm{d}\beta \geq t \right\}, \quad t \in \Bbb R,\quad Q \in \mathcal{Q},
\end{equation}
is law invariant. By \eqref{eq:trho-3}, it follows that
\begin{equation} \label{eq: trho-law inv-1}
\trho(X)= \exp \left( \sup_{Q \in \mathcal{Q}} \mathfrak{R} \left(\int_0^1 q_{\log X} (\beta) \, q_{\frac{\mathrm{d}Q}{\mathrm{d}P}} (\beta) \,\mathrm{d}\beta;Q \right) \right).
\end{equation}
Because it is easy to check that $q_{\log X} (\beta)= \log (q_X (\beta))$ for any $X \in L^{\infty}_{++}$ and any $\beta \in [0,1]$ (see also Lemma A.23 in F\"{o}llmer and Schied \cite{FS11}), \eqref{eq: trho-law inv-1} becomes
\begin{eqnarray}
\trho(X) &=& \sup_{Q \in \mathcal{Q}} \exp \left(\mathfrak{R} \left(\int_0^1 \log (q_{X} (\beta)) \, q_{\frac{\mathrm{d}Q}{\mathrm{d}P}} (\beta) \,\mathrm{d}\beta;Q \right) \right) \label{eq: lawinv-Rtt}\\
&=& \sup_{Q \in \mathcal{Q}} \mathcal{R} \left(\int_0^1 \log (q_{X} (\beta)) \, q_{\frac{\mathrm{d}Q}{\mathrm{d}P}} (\beta) \,\mathrm{d}\beta;Q \right), \notag
\end{eqnarray}
where $\mathcal{R}(t;Q)= \exp\left( \mathfrak{R}(t;Q)\right)$.
Furthermore, for any $t \in \Bbb R$ and $Q \in \mathcal{Q}$ it holds that
\begin{eqnarray*}
\mathcal{R}(t;Q) &=& \inf_{ Y \in L^{\infty}_{++}} \{ \trho(Y): \mathbb{E}_Q [\log (Y)] \geq t \} \\
&=&\inf_{ Y \in L^{\infty}_{++}} \left\{ \trho(Y): \int_0^1 q_{\log Y} (1-\beta) \, q_{\frac{\mathrm{d}Q}{\mathrm{d}P}} (\beta) \,\mathrm{d}\beta \geq t \right\} \\
&=&\inf_{ Y \in L^{\infty}_{++}} \left\{ \trho(Y): \int_0^1 \log (q_Y (1-\beta)) \, q_{\frac{\mathrm{d}Q}{\mathrm{d}P}} (\beta) \,\mathrm{d}\beta \geq t \right\},
\end{eqnarray*}
where the second equality can be proved proceeding as in \cite{CMMM11}.
\end{proof}

Note that, by
\begin{equation*}
\mathcal{R}(t;Q)=\exp\left(\mathfrak{R}(t;Q) \right)=R(\exp t;Q),
\end{equation*}
\eqref{eq: dual repres rho-tilde-law inv} can be rewritten as
\begin{equation} \label{eq: dual repres rho-tilde-law inv-2}
\trho(X)= \sup_{Q \in \mathcal{Q}} R \left(\exp\left(\int_0^1 \log (q_X (\beta)) \, q_{\frac{\mathrm{d}Q}{\mathrm{d}P}} (\beta) \,\mathrm{d}\beta\right) ;Q \right).
\end{equation}

\begin{proposition}
\begin{itemize}
\item[(a)] $\trho \colon L^{\infty}_{++} \to [0, \infty]$ is monotone, quasi-logconvex, continuous from below, law invariant and star-shaped if and only if it the associated $\mathcal{R}$ is law invariant and multiplicatively expansive
in the first coordinate.

\item[(b)] $\trho \colon L^{\infty}_{++} \to [0, \infty]$ is monotone, quasi-logconvex, continuous from below, law invariant and positively homogeneous if and only if the associated $\mathcal{R}$ is law invariant and multiplicatively homogeneous
in the first coordinate.
\end{itemize}
\end{proposition}

\begin{proof}
(a) and (b) follow by Theorems~\ref{thm: dual repres rho-tilde -star-shaped}, \ref{thm: dual repres rho-tilde -PH} and \ref{thm: law invar-repres}. 
Alternatively, they can be proved directly thanks to \eqref{eq: definition-tR-law inv}.
\end{proof}

Two alternative representations of law invariant quasi-logconvex risk measures are established in the following results: the former in terms of Orlicz premia, the latter in terms of AR@R.

\begin{theorem} \label{thm: law invar-repres-Orlicz}
\textbf{[Representation of law invariant quasi-logconvex risk measures via Orlicz premia]}

Let $(\Omega, \mathcal{F},P)$ be a non-atomic probability space.

If $\trho \colon L^{\infty}_{++} \to [0, \infty]$ is monotone, quasi-logconvex, continuous from below and law invariant,
then it has the following representation:
\begin{equation} \label{eq: repres rho-tilde-law inv-Orlicz}
\trho(X)= \sup_{Q \in \mathcal{Q}} R \left(\sup_{\tilde{Q}\ll P: \frac{\mathrm{d}\tilde{Q}}{\mathrm{d}P} \overset{d}{\sim} \frac{\mathrm{d}Q}{\mathrm{d}P}} H_{0, \tilde{Q}} (X);Q \right)
= \sup_{Q \in \mathcal{Q}} \sup_{\tilde{Q}\ll P: \frac{\mathrm{d}\tilde{Q}}{\mathrm{d}P} \overset{d}{\sim} \frac{\mathrm{d}Q}{\mathrm{d}P}} R \left( H_{0, \tilde{Q}} (X);Q \right).
\end{equation}
\end{theorem}

\begin{proof}
By \eqref{eq: lawinv-Rtt},
\begin{equation} \label{eq: trho-law inv-R-q}
\trho(X)= \sup_{Q \in \mathcal{Q}} \exp \left( \mathfrak{R} \left(\int_0^1 q_{\log X} (\beta) \, q_{\frac{\mathrm{d}Q}{\mathrm{d}P}} (\beta) \,\mathrm{d}\beta;Q \right) \right).
\end{equation}
It then follows that
\begin{eqnarray*}
\trho(X)&=& \sup_{Q \in \mathcal{Q}} \exp \left( \mathfrak{R} \left( \sup_{\tilde{Q}\ll P: \frac{\mathrm{d}\tilde{Q}}{\mathrm{d}P} \overset{d}{\sim} \frac{\mathrm{d}Q}{\mathrm{d}P}} \mathbb{E}_{\tilde{Q}}[\log X];Q \right) \right) \\
&=& \sup_{Q \in \mathcal{Q}} R \left( \exp\left(\sup_{\tilde{Q}\ll P: \frac{\mathrm{d}\tilde{Q}}{\mathrm{d}P} \overset{d}{\sim} \frac{\mathrm{d}\tilde{Q}}{\mathrm{d}P}} \mathbb{E}_{\tilde{Q}}[\log X] \right);Q \right) \\
&=& \sup_{Q \in \mathcal{Q}} R \left(\sup_{\tilde{Q}\ll P: \frac{\mathrm{d}\tilde{Q}}{\mathrm{d}P} \overset{d}{\sim} \frac{\mathrm{d}Q}{\mathrm{d}P}} \exp\left(\mathbb{E}_{\tilde{Q}}[\log X] \right);Q \right)  \\
&=& \sup_{Q \in \mathcal{Q}} R \left(\sup_{\tilde{Q}\ll P: \frac{\mathrm{d}\tilde{Q}}{\mathrm{d}P} \overset{d}{\sim} \frac{\mathrm{d}Q}{\mathrm{d}P}} H_{0, \tilde{Q}} (X);Q \right)  \\
&=& \sup_{Q \in \mathcal{Q}} \sup_{\tilde{Q}\ll P: \frac{\mathrm{d}\tilde{Q}}{\mathrm{d}P} \overset{d}{\sim} \frac{\mathrm{d}Q}{\mathrm{d}P}} R \left( H_{0, \tilde{Q}} (X);Q \right),
\end{eqnarray*}
where the first equality is due to Proposition~14 of Kusuoka \cite{K01}, while the second is due to the relation  $\exp\left(\mathfrak{R}(t;Q)\right) = R(\exp(t);Q)$.
\end{proof}

Finally, the following result provides a dual representation of $\trho$ in terms of AR@R, in line with similar results by Kusuoka \cite{K01} and Frittelli and Rosazza Gianin \cite{FR05} for coherent and convex risk measures, respectively.

\begin{corollary}
If $\trho \colon L^{\infty}_{++} \to [0, \infty]$ is monotone, quasi-logconvex, continuous from below and law invariant,
then it has the following dual representation:
\begin{equation} \label{eq: dual repres rho-tilde-law inv-3}
\trho(X)= \sup_{Q \in \mathcal{Q}} \mathcal{R} \left(\int_{[0,1)} \log (AR@R_{\alpha} (X)) \, m_Q( \mathrm{d} \alpha);Q \right),
\end{equation}
where $m_Q$ is a probability measure on $[0,1)$ depending on $Q$ and AR@R is defined in Example~\ref{ex: ARaR}.
\end{corollary}

\begin{proof}
By Theorem \ref{thm: law invar-repres}, $\trho$ has the following dual representation:
\begin{eqnarray}
\trho(X)&=& \sup_{Q \in \mathcal{Q}} \mathcal{R} \left(\int_0^1 \log (q_X (\beta)) \, q_{\frac{\mathrm{d}Q}{\mathrm{d}P}} (\beta) \,\mathrm{d}\beta;Q \right) \notag\\
&=& \sup_{Q \in \mathcal{Q}} \mathcal{R} \left(\int_0^1 q_{\log X} (\beta) \, q_{\frac{\mathrm{d}Q}{\mathrm{d}P}} (\beta) \,\mathrm{d}\beta;Q \right) \label{eq: law inv-quant-1}
\end{eqnarray}
where $q_{Z}(\beta)$ is any $\beta$-quantile of $Z$.

Proceeding similarly as in Theorem~4 of Kusuoka \cite{K01} (see also Theorem~7 of Frittelli and Rosazza Gianin \cite{FR05}),
we consider now the smallest quantiles $q^-$ and
\begin{equation*}
\hat{q}_{\frac{\mathrm{d}Q}{\mathrm{d}P}}(y)= \left\{
\begin{array}{rl}
\lim_{y \to 0^+} q_{\frac{\mathrm{d}Q}{\mathrm{d}P}}^-(y),& \quad y \leq 0; \\
 q_{\frac{\mathrm{d}Q}{\mathrm{d}P}}^-(y),& \quad 0 < y \leq 1; \\
1,& \quad y >1. \\
\end{array}
\right.
\end{equation*}
It follows that
\begin{eqnarray}
\int_0^1 q_{\log X} (\beta) \, q_{\frac{\mathrm{d}Q}{\mathrm{d}P}} (\beta) \,\mathrm{d}\beta &=& \int_0^1 q^-_{\log X} (\beta) \, q^-_{\frac{\mathrm{d}Q}{\mathrm{d}P}} (\beta) \,\mathrm{d}\beta \notag\\
&=& \int_0^1 q^-_{\log X} (\beta) \, \hat{q}_{\frac{\mathrm{d}Q}{\mathrm{d}P}} (\beta) \,\mathrm{d}\beta \notag\\
&=& \int_{[0,1)} \left[\int_{\alpha}^1  q^-_{\log X} (\beta) \,\mathrm{d}\beta \right] \mathrm{d} \hat{q}_{\frac{\mathrm{d}Q}{\mathrm{d}P}} (\alpha) \notag\\
&=& \int_{[0,1)} \left[\frac{1}{1-\alpha} \int_{\alpha}^1  q^-_{\log X} (\beta) \,\mathrm{d}\beta \right] m_Q (\mathrm{d} \alpha) \notag\\
&=& \int_{[0,1)} \log (AR@R_{\alpha} (X)) \, m_Q (\mathrm{d} \alpha), \label{eq: law inv-quant-2}
\end{eqnarray}
where $m_Q (\mathrm{d} \alpha) \triangleq (1-\alpha) \,\mathrm{d} \hat{q}_{\frac{\mathrm{d}Q}{\mathrm{d}P}} (\alpha)$ is a probability measure on $[0,1)$ depending on $Q$.

Eqns.~\eqref{eq: law inv-quant-1} and \eqref{eq: law inv-quant-2} then imply \eqref{eq: dual repres rho-tilde-law inv-3}.
\end{proof}

\setcounter{equation}{0}

\section{Examples}\label{sec:examples}

In this section, we provide some illustrative examples of (law-invariant) quasi-logconvex risk measures.

\begin{example}
Consider the monetary risk measure $\rho$ with $\mathfrak{R}(t,Q)=t$.
It follows that $\rho$ is law invariant, coherent and continuous from below.

The return counterpart $\trho$ of $\rho$ is then logcoherent, law invariant, and given by
\begin{align*}
\trho(X)&=\sup_{Q\in\mathcal{Q}}\exp \left(\mathfrak{R}\left(\mathbb{E}_{Q}\left[\log X\right]; Q\right)\right) \\
&=\sup_{Q\in\mathcal{Q}}\exp \left(\mathbb{E}_{Q}\left[\log X\right]\right) \\
&=\sup_{Q\in\mathcal{Q}}H_{0, Q} (X).
\end{align*}
$\trho$ is therefore a supremum of Orlicz premia over the set $\mathcal{Q}$.
\end{example}

\begin{example}
Consider now $\rho$ corresponding to $\mathfrak{R}(t,Q)=t-c(Q)$ for a given convex penalty function $c$.
It follows that $\rho$ is convex and continuous from below, while its return counterpart $\trho$ is logconvex and given by
\begin{align*}
\trho(X)&=\sup_{Q\in\mathcal{Q}}\exp \left(\mathfrak{R}\left(\mathbb{E}_{Q}\left[\log X\right]; Q\right)\right) \\
&=\sup_{Q\in\mathcal{Q}}\exp \left(\mathbb{E}_{Q}\left[\log X\right] - c(Q)\right) \\
&=\sup_{Q\in\mathcal{Q}} \{ \beta(Q) H_{0, Q} (X)\},
\end{align*}
with $\beta(Q)=e^{-c(Q)}$.
Furthermore, if $c(\cdot)$ is law invariant, then the same holds for $\rho$ and $\trho$.
$\trho$ can thus be interpreted as a supremum of discounted Orlicz premia over the set $\mathcal{Q}$.
\end{example}

\begin{example}
Let $\rho$ be associated to $\mathfrak{R}(t,Q)=\sup_{q\in[0,1]}\{qt-c(qQ)\}$ for a given convex penalty function $c$.
It follows that $\rho$ is convex and continuous from below, while its return counterpart $\trho$ is logconvex and given by
\begin{align*}
\trho(X)&=\sup_{Q\in\mathcal{Q}}\exp \left(\sup_{q\in[0,1]}\{q\mathbb{E}_{Q}\left[\log X\right]-c(qQ)\} \right) \\
&=\sup_{Q\in\mathcal{Q},q\in[0,1] }\exp \left(q\mathbb{E}_{Q}\left[\log X\right]-c(qQ) \right) \\
&=\sup_{Q\in\mathcal{Q}, q \in [0,1]} \{ e^{-c(qQ)} H_{0, Q} (X^q)\}.
\end{align*}
Furthermore, if $c(\cdot)$ is law invariant, then the same holds for $\rho$ and $\trho$.

In other words, $\trho$ can be written as a supremum of discounted Orlicz premia of $X^q$.
\end{example}

\begin{example}
Let $\rho$ correspond to $\mathfrak{R}(t,Q)=t \vee C$ for a given level $C \in \Bbb R$. 
It follows that $\mathfrak{R}(t,Q)$ is law invariant, increasing and continuous in $t$, hence $\rho$ is quasi-convex, law invariant and continuous from below.

Furthermore, the geometric counterpart $\trho$ of $\rho$ is quasi-logconvex, law invariant, and given by
\begin{align*}
\trho(X)&=\sup_{Q\in\mathcal{Q}}\exp \left(\mathbb{E}_{Q}\left[\log X\right] \vee C\right) \\
&=\sup_{Q\in\mathcal{Q}} \{\exp \left(\mathbb{E}_{Q}\left[\log X\right]\vee C \right)\} \\
&=\sup_{Q\in\mathcal{Q}} H_{0, Q} (X) \vee e^C.
\end{align*}
\end{example}

\begin{example}
Consider now $\rho$ with
\begin{equation*}
\mathfrak{R}(t,Q)=\left\{ \begin{array}{rl} \log a ,& \, t\leq a; \\ \log t,& \, t >a; \end{array}\right.
\end{equation*}
for a given $a \in (0,1)$. 
Law invariance, increasing monotonicity and continuity in $t$ of $\mathfrak{R}(t,Q)$ then implies that $\rho$ is then quasi-convex, law invariant and continuous from below.

It then follows that its geometric counterpart
\begin{align*}
\trho(X)&=\sup_{Q\in\mathcal{Q}}\exp \left(\mathfrak{R}\left(\mathbb{E}_{Q}\left[\log X\right]\right)\right) \\
&=\sup_{Q\in\mathcal{Q}} \{\mathbb{E}_{Q}\left[\log X\right]\vee a \} \\
&=\sup_{Q\in\mathcal{Q}} \mathbb{E}_{Q}\left[\log X\right] \vee a
\end{align*}
is quasi-logconvex and law invariant.
\end{example}

\setcounter{equation}{0}

\section{Applications}\label{sec:app}

In this section, we investigate the implications of quasi-logconvex risk measures in applications to portfolio choice and capital allocation, both supporting and further motivating the use of quasi-logconvex risk measures.

\subsection{Portfolio Choice}

Let $\trho \colon L^{\infty}_{++} \to [0, \infty]$ be a monotone, quasi-logconvex and continuous from below risk measure (that is not necessarily quasi-convex) and let $\rho \colon L^{\infty} \to [-\infty,\infty]$ be the associated arithmetic risk measure that is monotone, quasi-convex and continuous from below.

The portfolio choice problem for (quasi-)convex risk measures (see Ruszczynski and Shapiro \cite{RS06} and Mastrogiacomo and Rosazza Gianin \cite{MRG15}), adapted to risk assessment of portfolio returns rather than assessment of a portfolio's monetary value, is given by
\begin{equation} \label{eq: portfolio-rho}
\min_{\mathbf{w}
\in \mathcal{W}: \, \mathbb{E} [\sum_{i=1}^n X_i w_i ] \leq r } \, \rho\left(\sum_{i=1}^n X_i w_i \right),
\end{equation}
where $\mathcal{W} \triangleq \{\mathbf{w}=(w_1, \ldots, w_n) \in \mathbb{R}^n: \sum_{i=1}^n w_i=1; w_i\geq 0 \, \forall i=1,\ldots,n\}$, $w_i$ denotes the weight to be invested in risky asset $i$ with log return $X_i$, and $r \in \mathbb{R}$ stands for a target level. 
Note that the constraint set is formulated in terms of ``$\leq r$'' instead of ``$\geq r$'' because of the sign convention used in the paper (positive realizations of random variables represent losses).
For simplicity of notation, we will denote $C \triangleq \{\mathbf{w} \in \mathcal{W}: \mathbb{E} [\sum_{i=1}^n X_i w_i ] \leq r \}$.

By the one-to-one correspondence between $\rho$ and $\trho$ (see~\eqref{eq:trho-3} and~\eqref{eq:trho-4}), the portfolio choice problem~\eqref{eq: portfolio-rho} can be formulated in terms of the geometric counterpart as follows:
\begin{equation*}
\min_{\mathbf{w} \in C} \log \left(\trho\left(e^{\sum_{i=1}^n X_i w_i} \right) \right)= \log \left(\min_{\mathbf{w} \in C}  \trho\left(\prod_{i=1}^n e^{X_i w_i} \right)  \right).
\end{equation*}
Hence, the portfolio choice problem above reduces to
\begin{equation} \label{eq: portfolio-rho-tilde}
\min_{\mathbf{w} \in \mathcal{W}: \, \mathbb{E} [\sum_{i=1}^n \log (Y_i) w_i ] \leq r} \trho\left(\prod_{i=1}^n Y_i^{w_i} \right),
\end{equation}
with $Y_i= e^{X_i}$.\smallskip

The optimization problem \eqref{eq: portfolio-rho-tilde} can be seen as the multiplicative version of the classical portfolio choice problem (hence, in view of Section~\ref{sec:mot}, is natural and financially reasonable for geometric risk measures and rebalanced portfolios).

In the following, and by extending the results in \cite{RS06} and \cite{MRG15}, we will deal with two different (but related) problems: we start with the general problem
\begin{equation*}
\min_{Z \in C} \trho(F(Z)),
\end{equation*}
where $C$ is a convex subset of a normed vector space $\mathcal{Z}$ and $F$ is a suitable functional defined on $\mathcal{Z}$, to end up with the portfolio choice for $\trho$ in~\eqref{eq: portfolio-rho-tilde}.

\subsubsection{General optimization problem} \label{sec: general optim}

In this subsection, we let $\trho \colon L^{\infty}_{++} \to [0, \infty]$ be a monotone, quasi-logconvex and continuous from below risk measure (that is not necessarily quasi-convex).

We consider the general optimization problem
\begin{equation} \label{eq: general minimization}
\min_{Z \in C} \trho(F(Z)),
\end{equation}
where $C$ is a convex subset of a normed vector space $\mathcal{Z}$ and $F: \mathcal{Z} \to L^{\infty}$ is a concave and continuous functional
such that $F(C) \subseteq L^{\infty}_{++}$.
By Theorem~\ref{thm: dual repres rho-tilde}, the problem above becomes
\begin{equation} \label{eq: general minimization-trho-R}
\min_{Z \in C} \exp\left(\sup_{Q \in \mathcal{Q}} \mathfrak{R}(\mathbb{E}_{Q} [\log F(Z)];Q) \right).
\end{equation}
Furthermore, in the following we will assume that the supremum over $Q$ can be restricted over a convex set $\mathcal{Q}_0 \subseteq \mathcal{Q}$ and that
$L(X;Q) \triangleq \mathfrak{R}(\mathbb{E}_{Q} [X];Q)$ is
lower semi-continuous in $X$, while quasi-concave and upper semi-continuous in $Q$. 
Note that $L$ is automatically quasi-convex in $X$ because of increasing monotonicity of $\mathfrak{R}(t;Q)$ in $t$ (implying also quasi-convexity of $\mathfrak{R}(t;Q)$ in $t$).
%
\medskip

The following result provides a sufficient condition for $Z$ to be an optimal solution of \eqref{eq: general minimization}. 
For the reader's convenience, we recall (see \cite{PZ00}) the definitions of the normal cone and of the star and of the Greenberg-Pierskalla subdifferentials of a functional $f: \mathcal{Y} \to \bar{\mathbb{R}}$:
\begin{eqnarray*}
N(C,Y) &\triangleq& \{Y^* \in \mathcal{Y}^*: \langle Y^*, U-Y \rangle \leq 0 \text{ for any } U \in C \}, \, Y \in \mathcal{Y}, \\
\partial ^{(*)} f(\bar{Y}) &\triangleq& \{Y^* \in \mathcal{Y}^*: \langle Y^*, Y-\bar{Y} \rangle \leq 0 \text{ for any } Y \in \mathcal{Y} \text{ s.t. } f(Y) < f(\bar{Y}) \}, \\
\partial ^{GP} f(\bar{Y}) &\triangleq& \{Y^* \in \mathcal{Y}^*: \langle Y^*, Y-\bar{Y} \rangle < 0 \text{ for any } Y \in \mathcal{Y}\text{ s.t. } f(Y) < f(\bar{Y}) \},
\end{eqnarray*}
where $\mathcal{Y}^*$ denotes the dual space of $\mathcal{Y}$ and $\langle \cdot, \cdot \rangle$ the pairing functional on $\mathcal{Y} \times \mathcal{Y}^*$.

\begin{proposition} \label{prop: suff cond- opt problem 1}
Let $\mathcal{Z}= \mathbb{R}^n$, $\mathcal{Q}_0$ be weakly compact, $C$ be a closed subset of $\mathcal{Z}$, and let $\trho$ and the associated $\mathfrak{R}$ satisfy the assumptions above. 
Let $(\bar{Z},\bar{Q}) \in \mathcal{Z} \times \mathcal{Q}_0$ be given and assume that $\bar{Z}$ is not a local minimizer of $\mathfrak{R}(\mathbb{E}_{\bar{Q}} [\log F(\cdot)];\bar{Q})$.

If $(\bar{Z},\bar{Q}) \in \mathcal{Z} \times \mathcal{Q}_0$ satisfies
\begin{equation} \label{eq: optimality condit-trho}
\partial ^{(*)} \mathbb{E}_{\bar{Q}} [\log (F(\bar{Z}))] \cap \left(-N(C,\bar{Z}) \right) \neq \{0\} \quad \text{and} \quad \bar{Q} \in \partial^{GP} \trho(F(\bar{Z})),
\end{equation}
then $(\bar{Z},\bar{Q})$ is an optimal solution of \eqref{eq: general minimization-trho-R}.
\end{proposition}

\begin{proof}
By the one-to-one correspondence between $\rho$ and $\trho$ (see \eqref{eq:trho-3} and \eqref{eq:trho-4}), the optimization problem \eqref{eq: general minimization} can be rewritten as
$\exp\left(\min_{Z \in C} \rho(\log F(Z))\right)$. In order to solve the initial problem \eqref{eq: general minimization}, it is therefore enough to focus on
\begin{equation*}
\min_{Z \in C} \rho(\log F(Z)),
\end{equation*}
while for \eqref{eq: general minimization-trho-R} we can restrict attention to
\begin{equation} \label{eq: general minimization-rho-R}
\min_{Z \in C} \sup_{Q \in \mathcal{Q}_0} \mathfrak{R}(\mathbb{E}_{Q} [\log F(Z)];Q),
\end{equation}
where $\mathcal{Q}$ has been replaced by $\mathcal{Q}_0$ because of the assumptions on $\trho$.

Since $\bar{F}(x) \triangleq \log (F(x))$ remains concave and continuous, Theorem~4 of \cite{MRG15} guarantees that if
\begin{equation*} 
\partial ^{(*)} \mathbb{E}_{\bar{Q}} [\bar{F}(\bar{Z})] \cap \left(-N(C,\bar{Z}) \right) \neq \{0\} \quad \text{and} \quad \bar{Q} \in \partial^{GP} \rho(\bar{F}(\bar{Z})),
\end{equation*}
then $(\bar{Z},\bar{Q})$ is an optimal solution of \eqref{eq: general minimization-rho-R}, hence also of \eqref{eq: general minimization-trho-R}.
The thesis follows because
\begin{equation*}
\partial^{GP} \rho(\bar{F}(\bar{Z}))=\partial^{GP} \rho(\log (F(\bar{Z}))=\partial^{GP} \trho(F(\bar{Z})),
\end{equation*}
where the last equality is due to Prop.~2.7 of Penot and Z\v{a}linescu \cite{PZ00} since $\trho(F(\bar{Z}))=\exp\left(\rho(\log (F(\bar{Z}))) \right)$.
\end{proof}

We note that the assumption of $\bar{Z}$ not being a local minimizer is quite reasonable.
Whenever $\bar{Z}$ were a local minimizer for $\mathfrak{R}(\mathbb{E}_{\bar{Q}} [\log F(\cdot)];\bar{Q})$, indeed, it would hold that $\mathfrak{R}(\mathbb{E}_{\bar{Q}} [\log F(\bar{Z})];\bar{Q}) \leq \mathfrak{R}(\mathbb{E}_{\bar{Q}} [\log F(Z)];\bar{Q})$ for any $Z$ in a ball $B_{\varepsilon}(\bar{Z})=\{Z \in \Bbb R^n: \Vert Z- \bar{Z}\Vert <\varepsilon \}$ for some $\varepsilon >0$. 
Nevertheless, any $Z_{\varepsilon}$---close enough to $\bar{Z}$ and obtained by $\bar{Z}$ as $Z_{\varepsilon}=(\bar{Z}_1-\frac{\varepsilon}{2n}, \ldots,\bar{Z}_n -\frac{\varepsilon}{2n})$---belongs to $B_{\varepsilon}({\bar{Z}})$ and satisfies
\begin{equation*}
\mathfrak{R}(\mathbb{E}_{\bar{Q}} [\log F(Z_{\varepsilon})];\bar{Q})<\mathfrak{R}(\mathbb{E}_{\bar{Q}} [\log F(\bar{Z})];\bar{Q}),
\end{equation*}
if $\log F$ and $\mathfrak{R}(\cdot;\bar{Q})$ are strictly increasing. Hence, in this case, $\bar{Z}$ could not be a local minimizer for $\mathfrak{R}(\mathbb{E}_{\bar{Q}} [\log F(\cdot)];\bar{Q})$.

\begin{remark}
Note that the geometric nature of quasi-logconvex risk measures emerges from
\begin{align*}
&\partial ^{(*)} \mathbb{E}_{\bar{Q}} [\log (F(\bar{Z}))] \\
&= \left\{Z^* \in \mathbb{R}^n: \langle Z^*, Z-\bar{Z} \rangle  \leq 0 \text{ for any } Z \in \mathbb{R}^n \text{ s.t. } \mathbb{E}_{\bar{Q}} \left[\log \left( \frac{F(Z)}{F(\bar{Z})}\right)\right] <0\right\},
\end{align*}
and from a comparison between (the conditions on) the optimal solutions of the general optimization problem.

For a quasi-logconvex risk measure $\trho \colon L^{\infty}_{++} \to [0, \infty]$ satisfying the respective hypotheses, Proposition~\ref{prop: suff cond- opt problem 1} guarantees that if $(\bar{Z},\bar{Q}) \in \mathcal{Z} \times \mathcal{Q}_0$ satisfies
\begin{equation*}
\partial ^{(*)} \mathbb{E}_{\bar{Q}} [\log (F(\bar{Z}))] \cap \left(-N(C,\bar{Z}) \right) \neq \{0\} \quad \text{and} \quad \bar{Q} \in \partial^{GP} \trho(F(\bar{Z})),
\end{equation*}
then $(\bar{Z},\bar{Q})$ is an optimal solution of \eqref{eq: general minimization}--\eqref{eq: general minimization-trho-R}.

If $\trho \colon L^{\infty} \to [-\infty, \infty]$ were quasi-convex, instead, then Theorem~4 of \cite{MRG15} guarantees that if
\begin{equation*}
\partial ^{(*)} \mathbb{E}_{\bar{Q}} [F(\bar{Z})] \cap \left(-N(C,\bar{Z}) \right) \neq \{0\} \quad \text{and} \quad \bar{Q} \in \partial^{GP} \trho(F(\bar{Z})),
\end{equation*}
then $(\bar{Z},\bar{Q})$ is an optimal solution of \eqref{eq: general minimization}.

Intuitively, the geometric, nature of quasi-logconvex risk measures appears evident in the condition on $\partial ^{(*)}$.
\end{remark}

\subsubsection{Portfolio choice---multiplicative version}

Let us now focus on the multiplicative version of the portfolio choice problem for $\trho$, i.e.,
\begin{equation} \label{eq: portfolio-rho-tilde-2}
\min_{\mathbf{w} \in \mathcal{W}: \, \mathbb{E} [\sum_{i=1}^n \log (Y_i) w_i ] \leq r} \, \trho\left(\prod_{i=1}^n Y_i^{w_i} \right).
\end{equation}

We state the following proposition:
\begin{proposition}
Let $\trho$ and $\mathfrak{R}$ satisfy the same assumptions as in Section \ref{sec: general optim} and assume that $\bar{\mathbf{w}}$ is not a local minimizer of $\mathfrak{R}(\mathbb{E}_{\bar{Q}} [\mathbf{w} \cdot \log \mathbf{Y}];\bar{Q})$.

If $(\bar{\mathbf{w}},\bar{Q}) \in \mathcal{W} \times \mathcal{Q}_0$ satisfies
\begin{equation} \label{eq: optimality condit-trho-2}
\partial ^{(*)} \mathbb{E}_{\bar{Q}} \left[\log \left(\prod_{i=1}^n Y_i^{\bar{w}_i} \right)\right] \cap \left(-N(C,\bar{\mathbf{w}}) \right) \neq \{0\} \quad \text{and} \quad \bar{Q} \in \partial^{GP} \trho \left(\prod_{i=1}^n Y_i^{\bar{w}_i}\right),
\end{equation}
then $(\bar{\mathbf{w}},\bar{Q})$ is an optimal solution of \eqref{eq: portfolio-rho-tilde-2}.
\end{proposition}

\begin{proof}
As underlined before, to solve \eqref{eq: portfolio-rho-tilde-2} it is enough to focus on
\begin{equation} \label{eq: portfolio-choice-rho}
\min_{\mathbf{w} \in \mathcal{W}: \, \mathbb{E} [\sum_{i=1}^n \log (Y_i) w_i ] \leq r} \rho\left(\sum_{i=1}^n w_i \log Y_i \right)=\min_{\mathbf{w} \in \mathcal{W}: \, \mathbb{E} [F(\mathbf{w})] \leq r} \rho\left(F(\mathbf{w}) \right),
\end{equation}
with $F(\mathbf{w})= \mathbf{w} \cdot \log \mathbf{Y}$.

By Theorem~4 in \cite{MRG15}, if $(\bar{\mathbf{w}},\bar{Q}) \in \mathcal{W} \times \mathcal{Q}_0$ satisfies
\begin{equation*} 
\partial ^{(*)} \mathbb{E}_{\bar{Q}} \left[\sum_{i=1}^n \log (Y_i) \bar{w}_i \right] \cap \left(-N(C,\bar{\mathbf{w}}) \right) \neq \{0\} \quad \text{and} \quad \bar{Q} \in \partial^{GP} \rho \left(\sum_{i=1}^n \log (Y_i) \bar{w}_i \right),
\end{equation*}
then $(\bar{\mathbf{w}},\bar{Q})$ is an optimal solution of \eqref{eq: portfolio-choice-rho}.
The thesis then follows immediately because $\partial^{GP} \rho= \partial^{GP} \trho$ holds by Prop.~2.7 of Penot and Z\v{a}linescu \cite{PZ00}.
\end{proof}


\subsubsection{Efficient frontier}

We recall that the efficient frontier associated to the portfolio choice problem \eqref{eq: portfolio-rho} corresponds to
\begin{equation} \label{eq: efficient frontier-rho}
r \in \mathbb{R}^+ \mapsto \rho( \bar{\mathbf{w}} (r) \cdot \mathbf{X}),
\end{equation}
where $\bar{\mathbf{w}} (r) \in \Bbb R^n$ denotes the optimal portfolio choice when the return target level is $r$.
We note that, in our setting, the efficient frontier is decreasing (because of our convention on the signs). 
Furthermore, it is well known that the efficient frontier is convex for a convex risk measure $\rho$ (see \cite{BLS04}) and quasi-convex for a quasi-convex risk measure (see \cite{MRG15}).

We now focus momentarily on the following portfolio choice problem for a quasi-logconvex risk measure $\trho$:
\begin{equation} \label{eq: portfolio-general minimization}
\min_{\mathbf{w}
\in \mathcal{W}: \, \mathbb{E} [\sum_{i=1}^n X_i w_i ] \leq r } \, \trho\left(\sum_{i=1}^n X_i w_i \right),
\end{equation}
that is a particular case of the general problem \eqref{eq: general minimization}, with $\mathcal{Z}= \mathbb{R}^n$, $F(\mathbf{w})=\mathbf{w} \cdot \mathbf{X}$,
where $w_i$ denotes the weight to be invested in the risky asset $X_i$ and $r \in \mathbb{R}$ stands for a target level.

Consider now the efficient frontier associated to \eqref{eq: portfolio-general minimization}, defined as
\begin{equation} \label{eq: efficient frontier-trho-general}
r \in \mathbb{R}^+ \mapsto \trho( \bar{\mathbf{w}} (r) \cdot \mathbf{X}),
\end{equation}
where $\bar{\mathbf{w}} (r) \in \Bbb R^n$ denotes the optimal portfolio choice when the return target level is $r$.
It can easily be checked that the efficient frontier above is decreasing (again because of our convention on the signs). However, for quasi-logconvex return measures, it fails to satisfy convexity, quasi-convexity and quasi-logconvexity, in general. 
The main reason is that, indeed, this formulation of the efficient frontier does not make fully sense for quasi-logconvex risk measures $\trho$ because of their multiplicative nature.

For this reason, in the following, we consider a more general definition of efficient frontier associated to the following optimization problem:
\begin{equation} \label{eq: general minimization Cr}
\min_{\mathbf{w} \in C_r} \trho(F(\mathbf{w})),
\end{equation}
where $C_r \subseteq \mathbb{R}^n$ is a convex subset depending on a target level $r \in \mathbb{R}^+$ and $F: \mathbb{R}^n \to L^{\infty}$ is a concave and continuous functional with $F(C_r) \subseteq L^{\infty}_{++}$, and define the generalized efficient frontier as
\begin{equation} \label{eq: efficient frontier-trho-general Cr}
r \in \mathbb{R}^+ \mapsto \trho( F(\bar{\mathbf{w}} (r))),
\end{equation}
where $\bar{\mathbf{w}} (r) \in \Bbb R^n$ denotes the optimal portfolio choice of \eqref{eq: general minimization Cr} whose existence is guaranteed under suitable assumptions (see Proposition~\ref{prop: suff cond- opt problem 1}).

It easily follows that if $(C_r)_{r \in \mathbb{R}^+}$ is an increasing family of subsets of $\mathbb{R}^n$ in $r$, then the generalized efficient frontier is decreasing in $r$.

In the following, $\mathbf{w}_1 \mathbf{w}_2$ will denote improperly a vector in $\mathbb{R}^n$ whose components are given by the componentwise product of $\mathbf{w}_1$ and $\mathbf{w}_2$.

\begin{proposition}
Assume that both the following hypothesis hold: (a) $\trho \circ F$ is quasi-logconvex; and (b) the family $(C_r)_{r \in \mathbb{R}^+}$ is quasi-logconvex, i.e., $\mathbf{w}_1 \in C_{r_1}, \mathbf{w}_2 \in C_{r_2}$ implies that $\mathbf{w}_1^{\alpha} \mathbf{w}_2^{1-\alpha} \in C_{r_{\alpha}}$ with $r_{\alpha}=r_1 ^{\alpha} r_2^{1-\alpha}$ for any $\alpha \in [0,1]$.

Then, the generalized efficient frontier is quasi-logconvex.
\end{proposition}

\begin{proof}
Let $r_1, r_2 \in \mathbb{R}^+$ be arbitrarily fixed, let $\bar{\mathbf{w}}_1=\bar{\mathbf{w}}(r_1) \in C_{r_1}, \bar{\mathbf{w}}_2=\bar{\mathbf{w}}(r_1) \in C_{r_2}$ be the corresponding optimal solutions of \eqref{eq: general minimization Cr} and denote $r_{\alpha}=r_1 ^{\alpha} r_2^{1-\alpha}$ for any $\alpha \in [0,1]$.

It then holds that
\begin{equation*}
\trho(F(\bar{\mathbf{w}} (r_{\alpha}))) \leq \trho(F(\bar{\mathbf{w}}_1^{\alpha} \bar{\mathbf{w}}_2^{1-\alpha})))\leq \trho(F(\bar{\mathbf{w}}_1)) \vee \trho(F(\bar{\mathbf{w}}_2)),
\end{equation*}
where the first inequality is due to assumption~(b) guaranteeing that $\bar{\mathbf{w}}_1^{\alpha} \bar{\mathbf{w}}_2^{1-\alpha} \in C_{r_{\alpha}}$, while the last one is due to~(a). 
Quasi-logconvexity of the generalized efficient frontier then follows.
\end{proof}

Note that assumption~(a) above generalizes what happens in the convex (respectively quasi-convex) case where $\trho \circ F$ fulfills convexity (resp.~quasi-convexity) for a linear $F$. 
For instance, (a) is verified for a function $F$ satisfying $F(\mathbf{w}_1^{\alpha} \mathbf{w}_2^{1-\alpha})=(F(\mathbf{w}_1))^{\alpha} (F(\mathbf{w}_2))^{1-\alpha}$.

Assumption~(b), instead, generalizes logconvexity of acceptance sets. 
It is verified, for instance, for $C_r= \{\mathbf{w} \in \mathbb{R}^n_+: \mathbb{E}[\sum_{i=1}^n X_i \log w_i] \leq \log r\}$. 
For any $\mathbf{w}_1 \in C_{r_1}$, $\mathbf{w}_2 \in C_{r_2}$ and $\alpha \in [0,1]$, indeed,
\begin{eqnarray*}
\mathbb{E}\left[\sum_{i=1}^n X_i \log (w_{i,1}^{\alpha} w_{i,2}^{1-\alpha})\right] &=& \mathbb{E} \left[\sum_{i=1}^n X_i ( \alpha \log w_{i,1} + (1-\alpha) \log w_{i,2}) \right] \\
&=& \mathbb{E} \left[\alpha\sum_{i=1}^n X_i \log w_{i,1} + (1-\alpha) \sum_{i=1} ^n X_i \log w_{i,2} \right] \\
&\leq& \alpha \log r_1 + (1-\alpha) \log r_2= \log(r_{\alpha}),
\end{eqnarray*}
implying that $\mathbf{w}_1^{\alpha} \mathbf{w}_2^{1-\alpha} \in C_{r_{\alpha}}$.

\subsubsection{Efficient frontier---multiplicative version}

The arguments of the previous subsection motivate the study of a multiplicative version of the efficient frontier.
For this reason, we focus now on the following multiplicative version of efficient frontier associated to a quasi-logconvex return measure $\trho$ and the portfolio choice problem \eqref{eq: portfolio-rho-tilde-2}:
\begin{equation} \label{eq: efficient frontier-trho}
r \in \mathbb{R}^+ \mapsto \trho \left( \prod_{i=1}^n Y_i^{\bar{w}_i (r)} \right),
\end{equation}
where $\bar{\mathbf{w}} (r) \in \Bbb R^n$ denotes the optimal portfolio choice when the return target level is $r$.

\begin{proposition} \label{prop: efficient frontier}
If $\trho$ is a quasi-logconvex return risk measure and the portfolio choice problem \eqref{eq: portfolio-rho-tilde-2} admits (at least) a solution, then the efficient frontier is quasi-convex and decreasing.
\end{proposition}

\begin{proof}
Decreasing monotonicity of the efficient frontier is straightforward. It remains to check its quasi-convexity.

Let $r_1, r_2 \in \mathbb{R}^+$ be fixed arbitrarily and let $\bar{\mathbf{w}} (r_1), \bar{\mathbf{w}} (r_2) \in \Bbb R^n$ be the corresponding optimal solutions. Set $r_{\alpha}\triangleq\alpha r_1 + (1-\alpha)r_2$ for $\alpha \in [0,1]$.

It then follows that
\begin{eqnarray*}
\trho\left(\prod_{i=1}^n Y_i^{\bar{w}_i (\alpha r_1 + (1-\alpha)r_2)} \right) &\leq& \trho\left(\prod_{i=1}^n Y_i^{\alpha \bar{w}_i (r_1) + (1-\alpha) \bar{w}_i (r_2)} \right) \\
&=& \trho\left(\left(\prod_{i=1}^n Y_i^{ \bar{w}_i (r_1)}\right)^{\alpha} \left(\prod_{i=1}^n Y_i^{ \bar{w}_i (r_2)}\right)^{1-\alpha}  \right) \\
&\leq & \trho\left(\prod_{i=1}^n Y_i^{ \bar{w}_i (r_1)}\right) \vee \trho\left( \prod_{i=1}^n Y_i^{ \bar{w}_i (r_2)}  \right),
\end{eqnarray*}
where the first inequality holds because $\alpha \bar{\mathbf{w}} ( r_1) + (1-\alpha) \bar{\mathbf{w}} (r_2)$ belongs to $\mathcal{W}$ and satisfies the constraint $\mathbb{E} [\sum_{i=1}^n \log (Y_i) w_i] \leq r_{\alpha}$, while the last one follows from quasi-logconvexity of $\trho$.
\end{proof}

\subsection{Capital allocation}

\textit{Capital allocation problem for monetary or arithmetic risk measures.}

The standard capital allocation problem for monetary or arithmetic risk measures deals with suitably decomposing the risk assessment $\rho(X)$ of the whole/aggregate position $X$ into those of the different sub-units $X_1, \ldots,X_n$ of $X$, where $X=X_1+\ldots+X_n$. 
In particular, let $\Lambda(X_i,X)$ represent the capital to be allocated to the sub-unit $X_i$ of $X$. 
Heuristically, $\Lambda(X_i,X)$ can be seen as the capital allocated to $X_i$ to make $X_i$ acceptable (in terms of its riskiness) as a sub-portfolio of $X$.
See, e.g., risk contribution with respect to acceptance families $\Lambda_{\mathcal{A}}(X_i,X)= \inf\{ m \in \Bbb R: m-X_i \in \mathcal{A}_X\}$ proposed in Canna et al. \cite{CCR20} and adapted here to our convention on signs, where $\mathcal{A}_X$ denotes the set of all positions that are acceptable as sub-units of $X$ and, under suitable conditions on the family $(\mathcal{A}_X)_X$, $\rho(X)= \Lambda_{\mathcal{A}}(X,X)$.

Among the different methods, the class of proportional capital allocation principles provides popular capital allocation rules (CARs), consisting in
\begin{equation*}
\Lambda(X_i,X)= \rho(X) \cdot \frac{\bar{\rho}(X_i)}{\sum_{j=1}^n \bar{\rho}(X_j)},
\end{equation*}
for some risk measure $\bar{\rho}$. 
An alternative approach, proposed recently by Mohammed et al. \cite{MFS21} in the context of the Geometric Tail Expectation (GTE), consists in replacing the proportion $\frac{\bar{\rho}(X_i)}{\sum_{j=1}^n \bar{\rho}(X_j)}$ of the riskiness of $X_i$ with respect to $X=X_1+\ldots+X_n$ by a term depending on the ratio $R_i\triangleq \frac{X_i}{X}$.

The approach proposed by \cite{MFS21} for proportional capital allocations motivates the use of return or geometric risk measures and corresponding CARs in a multiplicative version. 
Return risk measures of the form $\tilde{\rho}(X)= \exp\left(\mathbb{E}_{\tilde{P}}[\log X]\right)$ are also motivated by Bauer and Zanjani \cite{BZ16} by means of economic/actuarial arguments as the risk measures corresponding to and resulting from desired capital allocations (see also Mohammed et al. \cite{MFS21} for a discussion and for the study of Geometric Tail Expectation).

\medskip

\noindent \textit{Capital allocation problem for return or geometric risk measures.}

Since return or geometric risk measures refer to relative positions instead of absolute values, they are suitable for dealing with capital allocation rules in terms of relative/multiplicative positions. 
In other words, $\tilde{\Lambda}(X_i,X)$ should represent the capital to be allocated to $X_i$ (in terms of its relative position) to make $X_i$ acceptable as a sub-portfolio of $X$ (in terms of return or geometric risk measures). 
Translating the terminology of Canna et al. \cite{CCR20} to this setting, with the comments on acceptance sets in mind (see Section~\ref{sec:acc}), the relation between CARs of return or geometric risk measures and acceptance sets can be formulated as
\begin{equation*}
\tilde{\Lambda}_{\mathcal{B}}(X_i,X)= \inf\left\{ m \in \Bbb R: \frac{m}{X_i} \in \mathcal{B}_X\right\},
\end{equation*}
where $\mathcal{B}_X$ denotes the set of all positions that are acceptable as sub-units of $X$ and, under suitable conditions on the family $(\mathcal{B}_X)_X$, $\tilde{\rho}(X)= \tilde{\Lambda}_{\mathcal{B}}(X,X)$; see also Bellini et al. \cite{BLR18}, Section~3.

Note that the subdifferential CAR is compatible with the approach with acceptance sets.
For instance, for coherent monetary risk measures $\rho$ the subdifferential CAR consists in $\Lambda^{sub}(X_i,X)=\mathbb{E}_{Q_X}[X_i]$ with $Q_X$ being an optimal generalized scenario in the dual representation of $\rho$.
At the level of the corresponding return risk measures $\tilde{\rho}$,
\begin{equation*}
\tilde{\Lambda}^{sub}(X_i,X)=\exp \left( \mathbb{E}_{Q_X}[\log X_i] \right),
\end{equation*}
since $Q_X$ is an optimal generalized scenario in the dual representation of $\rho$ if and only if it is also an optimal generalized scenario in the dual representation of $\tilde{\rho}$.
Choosing $\mathcal{B}_X= \left\{ Z \in L^{\infty}_{++}: \exp\left(\mathbb{E}_{Q_X}\left[\log \left(\frac{1}{X} \right)\right] \right) \leq 1 \right\}$,
\begin{equation*}
\tilde{\Lambda}_{\mathcal{B}}(X_i,X)=\exp \left( \mathbb{E}_{Q_X}[\log X_i] \right)= \tilde{\Lambda}^{sub}(X_i,X),
\end{equation*}
and $\tilde{\rho}(X)=\tilde{\Lambda}_{\mathcal{B}}(X,X)$.

As discussed above, the approach proposed by \cite{MFS21} for proportional capital allocations (and, in particular, for GTE-allocations) fits to our setting, of which it occurs as a special case.  
The idea of GTE-allocations can be generalized to general geometric risk measures $\tilde{\rho}$ as
\begin{equation*}
\tilde{\Lambda}^*(X_i,X)=\tilde{\rho}(X)\cdot \Lambda_{\rho}\left(\frac{X_i}{X},X\right),
\end{equation*}
where $\Lambda_{\rho}$ stands for a CAR of the corresponding arithmetic risk measure $\rho$.
When dealing with relative positions, however, it seems to be more appropriate to consider a CAR of the form
\begin{equation*}
\tilde{\Lambda}(X_i,X)=\tilde{\rho}(X)\cdot \tilde{\Lambda}^{prop}\left(X_i,X\right),
\end{equation*}
with $\tilde{\Lambda}^{prop}\left(X_i,X\right)$ a proportional CAR depending on ratios $\frac{X_i}{X}$, rather than $\tilde{\Lambda}^*(X_i,X)$ with $\Lambda_{\rho}\left(\frac{X_i}{X},X\right)$ referring to absolute values.
For instance, for a quasi-logconvex risk measure $\tilde{\rho}$ one can consider
\begin{equation*}
\tilde{\Lambda}^{prop}(X_i,X)= \exp \left(\mathfrak{R} \left( \mathbb{E}_{Q_X} \left[\log \left(\frac{X_i}{X} \right)\right]; Q_X\right) \right),
\end{equation*}
where $Q_X$ is an optimal generalized scenario in the dual representation of $\tilde{\rho}(X)= \exp \left(\mathfrak{R} \left( \mathbb{E}_{Q_X} [\log X], Q_X\right) \right)$.
In this case, it is easy to check that
\begin{itemize}
  \item if $\mathfrak{R}(0,Q_X)=0$, then $\tilde{\Lambda}^{prop}(X,X)=1$ and $\tilde{\Lambda}(X,X)= \tilde{\rho}(X)$;
  \item if $\mathfrak{R}(c,Q_X)=c$ for any $c \in \Bbb R$, then $\tilde{\Lambda}^{prop}(\beta X,X)=\beta$ for any $\beta >0$, i.e., corresponding to the proportion to be allocated to $\beta X$;
  \item when $\mathfrak{R}(t,Q)=t$ (coherent case),
\begin{equation*}
\tilde{\Lambda}^{prop}(X_i,X)= \exp \left( \mathbb{E}_{Q_X} \left[\log \left(\frac{X_i}{X} \right)\right] \right) = \frac{\exp \left( \mathbb{E}_{Q_X} [\log X_i] \right)}{\exp \left( \mathbb{E}_{Q_X} [\log X] \right)}= \frac{\tilde{\Lambda}^{sub}(X_i,X)}{\tilde{\rho}(X)}.
\end{equation*}
\end{itemize}

\end{document}